\PassOptionsToPackage{usenames,dvipsnames}{xcolor}
\documentclass[sigconf, nonacm]{acmart}

\newcommand\vldbdoi{XX.XX/XXX.XX}
\newcommand\vldbpages{XXX-XXX}
\newcommand\vldbvolume{14}
\newcommand\vldbissue{1}
\newcommand\vldbyear{2020}
\newcommand\vldbauthors{\authors}
\newcommand\vldbtitle{\shorttitle} 
\newcommand\vldbavailabilityurl{https://github.com/gegeji/hrnn}
\newcommand\vldbpagestyle{plain} 

\usepackage{listings}
\usepackage{xcolor}
\usepackage{tikz}
\usetikzlibrary{positioning,fit,arrows.meta}
\usepackage{cmd}
\usepackage{graphicx}

\usepackage{float}
         
\usepackage{pgfplots}     
\pgfplotsset{compat=1.16}

\usepackage{amsmath,amsfonts}
\usepackage{algorithmic}
\usepackage{textcomp}
\usepackage{booktabs}
\usepackage[a-1b]{pdfx}
\usepackage{subcaption}
\usepackage{colortbl}
\usepackage{bbding}
\usepackage{braket}
\usepackage{epsfig,latexsym,graphics}
\usepackage[font={small}]{caption}
\usepackage{setspace}
\usepackage{enumitem}

\usepackage{cmd}
\usepackage[spacesep,definevectors]{easyvector} 
\usepackage{tabularx}
\usepackage{makecell}
\usepackage{multirow}
\usepackage{bbm} 
\usepackage{fontawesome5}
\usepackage[normalem]{ulem} 
\renewcommand{\gg}{\mathrel{>\!\!>}} 

\SetAlgoInsideSkip{0pt}

\SetAlgoSkip{smallskip}

\usepackage{ifthen}
\newcommand{\themacro}{themacro}

\newcommand{\KNN}{\kw{KNN}}
\newcommand{\HRNN}{\kw{HRNN}}
\newcommand{\HAMG}{\kw{HAMG}}
\newcommand{\RDT}{\kw{RDT}}
\newcommand{\SFT}{\kw{SFT}}

\newcommand{\HNSWG}{G_{\kw{HNSW}}}
\newcommand{\KNNG}{G_{\kw{KNN}}}

\begin{document}
\title{HRNN: A Hybrid Graph Index for Approximate Reverse $k$-Nearest Neighbor Search on High-Dimensional Vectors}

\author{Wenxuan Xia}
\affiliation{%
  \institution{HKUST (GZ)}
  \country{China}
}
\email{wxia248@connect.hkust-gz.edu.cn}

\author{Mingyu Yang}
\affiliation{%
  \institution{HKUST (GZ) \& HKUST}
  \country{China}
}
\email{myang250@connect.hkust-gz.edu.cn}

\author{Wentao Li}
\affiliation{%
\institution{University of Leicester}
  \country{United Kingdom}
}
\email{wl226@leicester.ac.uk}

\author{Wei Wang}
\affiliation{%
  \institution{HKUST (GZ) \& HKUST}
  \country{China}
}
\email{weiwcs@ust.hk}


\begin{abstract}
Reverse $k$-nearest neighbor (R$k$NN) search returns all data points that regard a query vector as one of their $k$-nearest neighbors ($k$NNs). Existing R$k$NN methods typically follow a \textbf{filter-and-verification framework}: vectors near the query vector are first collected as candidates and then verified against their $k$NN-radius (i.e., the distance to their $k$-th nearest neighbor). However, existing methods face two key limitations in high-dimensional spaces. First, nearby vectors often do not belong to the query's true R$k$NN set, resulting in excessive candidate expansion overhead. Second, existing methods compute $k$NN-radius online during verification, incurring substantial query-processing cost.

To address these limitations, we propose $\HRNN$, a hybrid graph index for approximate R$k$NN search. 
(1) Rather than directly treating nearby vectors as R$k$NN candidates, $\HRNN$ uses them as proxy points based on the \textbf{assumption} that a query's R$k$NN results can often be discovered through the R$k$NN results of its nearby vectors. 
(2) To reduce verification cost, $\HRNN$ \textbf{materializes} high-fidelity $k$NN-radius offline, eliminating expensive online reconstruction while preserving accuracy. $\HRNN$ combines a navigation graph, a ranked $\KNN$ graph, and reverse-neighbor lists into a hybrid index that supports efficient proxy retrieval, candidate generation, and $k$NN-radius access. 
We also develop efficient index construction and append-only maintenance algorithms. Extensive experiments show that $\HRNN$ consistently outperforms existing methods, achieving up to one order of magnitude higher throughput. Moreover, $\HRNN$ scales to datasets containing up to 10 million high-dimensional vectors while supporting efficient dynamic index maintenance.
\end{abstract}

\maketitle


\vspace{-0.5em}
\pagestyle{\vldbpagestyle}
\begingroup\small\noindent\raggedright\textbf{PVLDB Reference Format:}\\
\vldbauthors. \vldbtitle. PVLDB, \vldbvolume(\vldbissue): \vldbpages, \vldbyear.\\
\href{https://doi.org/\vldbdoi}{doi:\vldbdoi}
\endgroup
\begingroup
\renewcommand\thefootnote{}\footnote{\noindent
This work is licensed under the Creative Commons BY-NC-ND 4.0 International License. Visit \url{https://creativecommons.org/licenses/by-nc-nd/4.0/} to view a copy of this license. For any use beyond those covered by this license, obtain permission by emailing \href{mailto:info@vldb.org}{info@vldb.org}. Copyright is held by the owner/author(s). Publication rights licensed to the VLDB Endowment. \\
\raggedright Proceedings of the VLDB Endowment, Vol. \vldbvolume, No. \vldbissue\ %
ISSN 2150-8097. \\
\href{https://doi.org/\vldbdoi}{doi:\vldbdoi} \\
}\addtocounter{footnote}{-1}\endgroup

\vspace{-0.5em}

\ifdefempty{\vldbavailabilityurl}{}{
\vspace{.3cm}
\begingroup\small\noindent\raggedright\textbf{PVLDB Artifact Availability:}\\
The source code, data, and/or other artifacts have been made available at \url{\vldbavailabilityurl}.
\endgroup
}

\section{Introduction}
With the rapid development of large language models (LLMs) and intelligent agents~\cite{LLM-RAG-NIPS-2020,MQA-VLDB-2024-Wang}, modern data are increasingly represented as high-dimensional vectors, such as text embeddings, image features, and other multi-modal representations~\cite{RetRobust-2024-ICLR, SelfRAG-2024-ICLR, RankRAG-2024-neuips}. 
Given a dataset $D$ of vectors/data points, the \textbf{$k$-nearest neighbors ($k$NN)} of a data point are the $k$ points closest to it, where distance is typically measured using Euclidean distance. 
Conversely, the \textbf{reverse $k$-nearest neighbors (R$k$NN)} of a query point $q$ consist of all data points in $D$ that regard $q$ as one of their $k$-nearest neighbors.
Given a dataset $D$, the R$k$NN search problem aims to efficiently retrieve the R$k$NN set of a query vector $q$ from $D$.

The R$k$NN set of a query vector $q$ naturally reflects its influence over the dataset $D$, since each vector in the R$k$NN set regards $q$ as one of its $k$-nearest neighbors. 
Consequently, R$k$NN search has been widely used in applications that require influence analysis and reverse relationship discovery. 
For example, in spatial analytics, R$k$NN search can identify attractive facility locations by finding sites that appear among the $k$-nearest neighbors of a large number of customers \cite{InfluenceZone-2011-ICDE}.
In retrieval-augmented generation (RAG) systems, R$k$NN search can help identify influential knowledge chunks that are frequently referenced by many queries.
More broadly, R$k$NN search has been extensively studied in machine learning tasks that rely on reverse influence relationships among data points, including density-based clustering~\cite{RNN-Clustering-TKDE-2017, KR-DBSCAN-2021-ExpSys, RNNClustering-2019-InfoSys}, outlier detection~\cite{RKNN-Outlier-detection-TKDE-2014, OutlierDetection-2015-PatRecLet, RNN-Outlier-Detection-2025-PatAna}, and class-imbalance sampling~\cite{SMOTE-RkNN-2022-Elesevier, RNN-Sampling-2019-PatRecLetter}.

To process the R$k$NN search for a query $q$, an \textbf{intuitive} approach is to first compute the $k$NN set of every point $o$ in the dataset $D$ and obtain its \textbf{$k$NN-radius}, i.e., the distance between $o$ and its $k$-th nearest neighbor. 
Then, by definition, if the distance between $q$ and point $o \in D$ does not exceed the $k$NN radius of $o$, $o$ is identified as an R$k$NN of $q$.
However, this intuitive approach becomes prohibitively expensive when the dataset is large or the vector dimensionality is high. 
Although prior studies~\cite{RKNN-Road-Network-SSTD-2025, InfluenceZone-2011-ICDE} have explored spatial indexes such as R-trees to accelerate R$k$NN search, exact R$k$NN retrieval in high-dimensional vector spaces remains challenging. 
Thus, this paper focuses on approximate R$k$NN (i.e., AR$k$NN) search.

\stitle{Existing Solutions.}
Existing methods for AR$k$NN search generally follow a \textbf{filter-and-verification framework}. 
They first retrieve a superset of potential R$k$NN results for a query $q$ as candidates, and then verify each candidate using a process similar to the intuitive approach described above. 
Early methods mainly rely on spatial partitioning~\cite{TPL-2004-VLDB, TPL2-2007-Springer} or precomputed distance bounds~\cite{MRKNNCop-2006-SIGMOD} to prune unpromising vectors from the dataset. 
Yet, these approaches suffer from high preprocessing and indexing costs in high-dimensional spaces due to the curse of dimensionality~\cite{Curse-of-dim-1998-ACM}. 
More recent methods, including $\SFT$~\cite{SFT-2003}, $\RDT$~\cite{RDT-2017-VLDB}, and the state-of-the-art $\HAMG$~\cite{HAMG-ICDE-2024}, are based on the assumption that \textbf{the R$k$NNs of a query vector often overlap significantly with its nearby neighbors.} 
Thus, these methods first retrieve nearby neighbors as candidates, such as the top-$k'$ nearest neighbors of the query vector in $\SFT$~\cite{SFT-2003} and $\RDT$~\cite{RDT-2017-VLDB}, where $k' \gg k$, or the $k$-hop neighbors in graph-based methods such as $\HAMG$~\cite{HAMG-ICDE-2024}.
They then verify whether each candidate belongs to the R$k$NN set of the query by comparing the query distance against the candidate's $k$NN-radius.

\begin{figure}[t]
    \centering
    \includegraphics[width=0.74\linewidth]{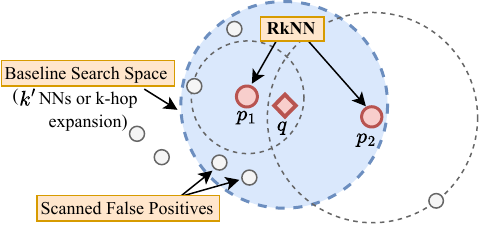}
    
    \vspace{-0.4cm}
    \caption{Illustration of the inefficiency of existing methods. The radius of the circles centered at points (e.g., $p_1$ and $p_2$) in the dataset $D$ denotes their $k$NN-radius. Given a query vector $q$, the goal is to retrieve its R$k$NN results ($p_1$ and $p_2$ in this example, as we set $k=1$). 
    Existing methods typically explore nearby neighbors, such as through top-$k'$ ($k' \gg k$) nearest-neighbor search or $k$-hop graph expansion, and then verify whether the retrieved candidates are true R$k$NN results. However, to reach $p_2$, which has a large $k$NN-radius due to its sparse neighborhood, these methods must traverse a large search region and examine many false positive candidates (gray points), resulting in substantial unnecessary computation.}
    \label{fig:motivation}
    \vspace{1em}
\end{figure}

\stitle{Limitations.}
Although recent methods achieve good performance by considering the properties of high-dimensional spaces, they still suffer from inherent limitations that restrict their efficiency.

\sstitle{Limitation 1: The mismatch between R$k$NN and nearby neighbors.}
Existing methods need to examine a large number of nearby neighbors to achieve high recall, because true R$k$NN results can be far away from the query vector in the embedding space (see Fig.~\ref{fig:motivation} and Fig.~\ref{fig:nonlocality-cdf}). 
For example, a point $o \in D$ may still regard query vector $q$ as one of its $k$-nearest neighbors if $o$ lies in a sparse region and therefore has a large $k$NN-radius. As a result, the query must search a large neighborhood region before reaching $o$, leading to excessive candidate expansion and unnecessary verification cost. This mismatch between R$k$NN results and nearby neighbors fundamentally limits the efficiency of existing methods.

\sstitle{Limitation 2: High verification cost.}
After generating candidate points, existing methods verify whether a candidate $o$ belongs to the final R$k$NN result by comparing the distance between $q$ and $o$ against the $k$NN-radius of $o$. 
Specifically, if $\delta(q,o)$ is no larger than the $k$NN-radius of $o$, then $q$ is regarded as one of the $k$-nearest neighbors of $o$, and thus $o$ is an R$k$NN of $q$. 
Yet, the $k$NN-radius of each candidate is typically computed online during query time~\cite{HAMG-ICDE-2024, RDT-2017-VLDB, Rdnn-tree-2001-ICDE}, as precomputing and storing these distances for all $k$ values incurs prohibitive preprocessing costs or requires rebuilding the index.
As a result, candidate verification introduces substantial computational overhead and becomes a major bottleneck in existing R$k$NN methods, as confirmed by our empirical studies.

\begin{table}[t]
    \caption{Comparison between $\HRNN$ and existing R$k$NN methods.}
    \label{tab:compare}
    \vspace{-0.4cm}
    \centering
    \small
    \setlength{\tabcolsep}{3pt}
    \renewcommand{\arraystretch}{1.05}
    \resizebox{0.9\columnwidth}{!}{%
    \begin{tabularx}{\columnwidth}{
      @{}>{\raggedright\arraybackslash}p{18mm}
      >{\raggedright\arraybackslash}X
      >{\raggedright\arraybackslash}X
      >{\raggedright\arraybackslash}X
      >{\raggedright\arraybackslash}X@{}
    }
    \toprule
     & $\HRNN$ (Ours) & $\HAMG$ & $\RDT$ & $\SFT$ \\
    \midrule
  
    Filter (Candidate Generation)
    & Reverse-neighbor lists of proxies 
    & $k$-hop graph traversal 
    & Incremental $k$NN search 
    & One-shot $k$NN search \\
    
    \midrule
    
    Verification 
    & Pre-materialized $k$NN-radius 
    & Online computation 
    & Online computation 
    & Online computation \\
    
    
    
    \bottomrule
    \end{tabularx}}\vspace{0.5em}
\end{table}

\stitle{Our Solution.}
To address the limitations of existing methods, we propose $\HRNN$, a \underline{H}ybrid graph index for approximate \underline{R}everse $k$-\underline{N}earest \underline{N}eighbor search. 
(1) To overcome the mismatch between R$k$NN results and nearby neighbors, rather than assuming that a query's nearby neighbors directly contain its R$k$NN results, we make a different assumption: \textbf{a query's R$k$NN results can often be discovered through the R$k$NN results of its nearby neighbors.} Based on this assumption, $\HRNN$ first retrieves the query's nearby neighbors (i.e., $k$NNs) as \textbf{proxy} points and then leverages their R$k$NN results to recover the query's R$k$NN results. To support this process, each point maintains a \textbf{reverse-neighbor list} that records the points in $D$ whose nearest-neighbor lists contain it, enabling efficient retrieval of its R$k$NN results for arbitrary $k$ when used as a proxy.
This design avoids excessive neighborhood expansion while maintaining high retrieval accuracy.
(2) To address the high verification cost of existing methods, $\HRNN$ eliminates the expensive online computation of $k$NN-radius during query processing. Instead, it materializes $k$NN-radius information offline, enabling efficient candidate verification at query time.

Achieving these goals is non-trivial, as the index must simultaneously support $k$NN retrieval for query vectors, reverse-neighbor-list access for proxy points, and fast $k$NN-radius lookup for arbitrary values of $k$. 
To this end, $\HRNN$ integrates two complementary graph structures.
First, it employs a navigation graph, such as $\HNSW$~\cite{HNSW-2020-TPAMI}, to retrieve the query's nearest neighbors as proxy points. 
Second, it maintains a ranked $\KNN$ graph~\cite{NN-Descent-WWW-2011} that stores nearest-neighbor relationships among data points. 
This graph provides direct access to pre-materialized $k$NN-radius for verification and serves as the basis for constructing reverse-neighbor lists, which enable efficient retrieval of proxy points' R$k$NN results for candidate generation.
By tightly integrating a navigation graph, a ranked $\KNN$ graph, and reverse-neighbor lists into a unified index, $\HRNN$ enables efficient R$k$NN query processing.
We further develop efficient index construction and insertion-only maintenance algorithms. A comparison between $\HRNN$ and existing methods is in Table~\ref{tab:compare}.

\stitle{Contributions.}
We summarize our main contributions as follows.

\sstitle{Problem analysis (\S\ref{sec:problem-analysis}).}
We identify two key shortcomings of existing methods for high-dimensional AR$k$NN search. 
First, true R$k$NN results can be far away from the query vector in the embedding space, and thus the mismatch between R$k$NN results and nearby neighbors leads to excessive false positive candidates and unnecessary computation. 
Second, the $k$NN-radius of each candidate is computed online, which incurs huge overhead during the verification stage.

\sstitle{Hybrid index design (\S\ref{sec:main-framework}).}
We propose $\HRNN$, a hybrid graph index for efficient AR$k$NN search. 
$\HRNN$ integrates a navigation graph, a ranked $\KNN$ graph, and reverse-neighbor lists derived from the ranked $\KNN$ graph. 
The navigation graph is used to efficiently retrieve nearby proxy points, while the reverse-neighbor lists are used to recover the R$k$NN candidates of the query vector. 
In addition, the ranked $\KNN$ graph materializes $k$NN-radius information for arbitrary values of $k$, enabling efficient candidate verification without expensive online $k$NN-radius computation.

\sstitle{Index construction and maintenance (\S\ref{sec:main-framework}).}
We develop efficient algorithms for $\HRNN$ index construction and maintenance. 
Specifically, we use the navigation graph as initialization seeds to improve the construction quality of the ranked $\KNN$ graph. 
We further propose efficient methods to construct reverse-neighbor lists and support insertion-only index maintenance under dynamic updates.

\sstitle{Extensive experimental studies (\S\ref{sec:experiment}).}
Extensive experiments on four real-world datasets show that $\HRNN$ improves the recall-throughput trade-off over strong R$k$NN baselines by up to one order of magnitude. 
Additional ablation, varying-$k$, and scalability experiments further validate the effectiveness of our $\HRNN$.

\ifthenelse{\isundefined{\themacro}}{
Due to space limitations, some proofs, algorithms, and experimental results are omitted and deferred to the technical report~\cite{technicalreport}.}{}

\begin{table}[t]
\caption{Notation summary.}\vspace{-1em}
\label{tbl:notation}
\footnotesize
\begin{tabular*}{\linewidth}{@{\extracolsep{\fill}} p{15mm} | p{70mm}}
\toprule
Notation & Description \\
\midrule


$\HNSWG$ & $\HNSW$ navigation graph \\
$\KNNG$ & Ranked $\KNN$ graph \\
$\KNNG[o]$ & Ranked $\KNN$ list of $o$, sorted by distance \\
$\mathbf{R}[o]$ & Reverse-neighbor list of $o$ \\
$r_k(o)$ & $k$NN radius (distance to the $k$-th nearest neighbor) of $o$ \\
$k$ & Target parameter for $k$NN or R$k$NN search\\
$\mathbf{K}$ & Number of stored neighbors per vertex in $\KNNG$ \\
$\Theta$ & Rank threshold for reverse candidate generation \\
$m$ & Number of proxy vectors retrieved from $\HNSWG$ \\
$\Theta_u$ & Rank threshold used during insertion maintenance \\
$m_u$ & Number of update proxies used during insertion maintenance \\
\bottomrule
\end{tabular*}
\vspace{0.5em}
\end{table}

\vspace{-0.5em}
\section{Preliminary}

We first formulate the approximate R$k$NN search problem, and then introduce the graph-based index that supports $k$NN retrieval.

\vspace{-0.5em}

\subsection{Problem Definition}
We first introduce the definitions of $k$NN and R$k$NN, and then formally define the studied problem.
Table~\ref{tbl:notation} summarizes the commonly used notations in this paper.

\vspace{-0.5em}

\begin{definition}[$k$-Nearest Neighbor ($k$NN)]
Given a dataset $D$ of points in Euclidean space $\mathbb{R}^d$, a $d$-dimensional query vector $q$, and an integer $k$, the $k$NN set $N_k(q)$ of $q$ is a subset of $D$ such that $|N_k(q)| = k$ and $\forall v \in N_k(q),\ \forall x \in D \setminus N_k(q),\ \delta(v,q)\le \delta(x,q)$.
where $\delta(v,q)$ denotes the Euclidean distance between $v$ and $q$.
\end{definition}

\vspace{-0.3em}

Note that $q$ may not belong to $D$. For each point $o \in D$, let $N_k(o)$ denote its $k$NN in $D$, and let its \textbf{$k$NN-radius} be defined as $r_k(o)=\delta(o,v)$, where $v \in N_k(o)$ is the $k$-th nearest neighbor of $o$.

\vspace{-0.3em}

\begin{definition}[Reverse $k$-Nearest Neighbor (R$k$NN)]
Given a dataset $D \subseteq \mathbb{R}^d$, a $d$-dimensional query vector $q$, and an integer $k$, the R$k$NN set of $q$ is defined as
$A_k(q)=\{o \in D \mid \delta(q,o)\le r_k(o)\}$.
\end{definition}

The R$k$NN set $A_k(q)$ contains all points $o \in D$ that regard $q$ as one of their $k$ nearest neighbors. Equivalently, the distance between $q$ and $o$ is no larger than the $k$NN-radius of point $o$.
Since exact R$k$NN search is computationally expensive in high-dimensional spaces, existing methods typically trade accuracy for efficiency by returning approximate results. We therefore formally define the problem of \textbf{approximate R$k$NN (AR$k$NN) search} as follows.

\vspace{-0.3em}

\begin{definition}[Approximate Reverse $k$-Nearest Neighbor (AR$k$NN) Search]\label{defi:arknn}
Given a dataset $D$, a query vector $q$, and an integer $k$, an AR$k$NN search returns a set $\hat A_k(q)\subseteq D$ as an approximation to the exact R$k$NN result set $A_k(q)$. Its accuracy is evaluated against $A_k(q)$ rather than by an $\epsilon$-relaxed membership predicate.
\end{definition}

Following~\cite{HAMG-ICDE-2024}, we use Recall@$k$ to evaluate the accuracy of AR$k$NN results against the exact R$k$NN ground truth.

\vspace{-0.3em}

\begin{definition}[Recall@$k$]\label{def:recall}
For a query $q$, let $A_k(q)$ denote the exact R$k$NN result set and let $\hat A_k(q)$ denote the approximate result set. The per-query Recall@$k$ is defined as
\[
\mathrm{Recall@}k(q)=
\begin{cases}
\frac{|A_k(q)\cap \hat A_k(q)|}{|A_k(q)|}, & |A_k(q)|>0,\\
1, & |A_k(q)|=0 \text{ and } |\hat A_k(q)|=0,\\
0, & |A_k(q)|=0 \text{ and } |\hat A_k(q)|>0.
\end{cases}
\]
For a query workload $Q$, Recall@$k$ is computed as the average of $\mathrm{Recall@}k(q)$ over all $q \in Q$.
\end{definition}

\begin{figure}
    \centering
    \includegraphics[width=0.55\linewidth]{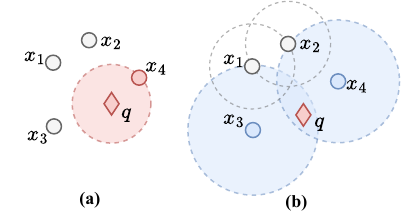}
     \vspace{-0.4cm}
    \caption{Illustrative examples of $k$NN and R$k$NN search.}
    \label{fig:knn-vs-rknn}
\end{figure}

\vspace{-1em}

\begin{example}
  Fig.~\ref{fig:knn-vs-rknn} shows $k$NN and R$k$NN search. 
  Let $x_1,\ldots,x_4$ be data points and $q$ the query, with $k=1$. In Fig.~\ref{fig:knn-vs-rknn}(a), the nearest neighbor of $q$ is $x_4$, and the (radius of the) circle centered at $q$ denotes its $k$NN-radius. In Fig.~\ref{fig:knn-vs-rknn}(b), the circles centered at $x_1,\ldots,x_4$ indicate their $k$NN-radius. Since $q$ lies within the $k$NN-radius of $x_3$ and $x_4$ but outside those of $x_1$ and $x_2$, the R$k$NN result set of $q$ is $\{x_3, x_4\}$.
\end{example}

\vspace{-1em}

\subsection{Graph-Based $k$NN Search}
We briefly introduce graph-based methods for high-dimensional $k$NN search, which serve as the foundation of the studied AR$k$NN search problem. Among existing high-dimensional $k$NN methods, graph-based approaches~\cite{HNSW-2020-TPAMI, NSG-VLDB-2019-Fu, TauMG-2023-SIGMOD, SSSG-2022-TPAMI, DiskANN-NIPS-2019, Graph-index-survey-VLDB-2021-Wang, RoarGraph-VLDB-2024-Chen, NN-Descent-WWW-2011, Flash-Graph-Index-2025-SIGMOD, Graph-tree-index-survey-2023-IEEE-DataEngBull, AlphaCG-2026-SIGMOD-Lu, Graph-Theory-2023-Indyk, Fast-Graph-2025-SIGMOD-Lu, GreedySearch-2011-IJCAI,PipeANN-2026-OSDI,VSAG-VLDB-2025-Mingyu,HVS-2021-VLDB} are particularly popular due to their state-of-the-art performance. These methods model the points in dataset $D$ as vertices in a graph and connect nearby vertices through directed or undirected edges. As a result, $k$NN search for a query vector $q$ becomes a graph traversal process over the constructed index.

\stitle{Ranked $\KNN$ Graph.}
We first introduce the ranked $\KNN$ graph, where $\mathbf{K}$ is a user-defined parameter chosen to be sufficiently large to support $k$NN search for arbitrary $k \le \mathbf{K}$.

\begin{definition}[Ranked $\KNN$ Graph]\label{def:knng}
Given a dataset $D \subset \mathbb{R}^d$ and an integer $\mathbf{K}$, the ranked $\KNN$ graph $\KNNG=(V,E)$ is a directed graph with vertex set $V=D$. Each vertex $o \in V$ stores an ordered \textbf{$\mathbf{K}$NN list} of its $\mathbf{K}$ nearest neighbors and their distances: $\KNNG[o] = \{(v_1,\delta(o,v_1)),\ldots,(v_{\mathbf{K}},\delta(o,v_{\mathbf{K}}))\}$, where $\delta(o,v_1) \le \ldots \le \delta(o,v_{\mathbf{K}})$.
A directed edge $o \rightarrow v_j$ exists if $v_j$ is the $j$-th nearest neighbor of $o$, i.e., $\KNNG[o,j]=v_j$.
For simplicity, we write $\KNNG[o,j]=v_j$, although the actual entry is $(v_j,\delta(o,v_j))$.
\end{definition}

The ranked $\KNNG$ captures nearest-neighbor relationships from the perspective of each point in $D$. Transposing these relationships yields reverse-neighbor lists, which record, for each vector, the points that regard it as one of their top-$\mathbf{K}$ nearest neighbors. Consequently, the R$k$NN results of any dataset vector can be obtained directly from its reverse-neighbor list for any $k \le \mathbf{K}$.

\begin{definition}[Reverse-Neighbor List]\label{defi:RR}
Given a ranked $\KNNG$, the reverse-neighbor list of a vector $o\in D$, denoted by $\mathbf{R}[o]$, is defined as $\mathbf{R}[o]
=
\{(v,j)\mid v\in D,\ \KNNG[v,j]=o,\ 1\le j\le \mathbf{K}\}$,
where $(v,j)$ indicates that $o$ is the $j$-th nearest neighbor of $v$.
\end{definition}

\begin{algorithm}[t]
\small
\caption{Ranked $\KNN$ Graph Construction}
\label{algo:knng}

\KwIn{Dataset $D$, integer $\mathbf{K}$}
\KwOut{Ranked $\KNN$ graph $\KNNG$}

Initialize $\KNNG[o]$ with $\mathbf{K}$ random points for eacn $o \in D$\;

\Repeat{no neighbor list changes}{
    Create $\mathbf{R}[o] \gets \{x \mid o \in \KNNG[x]\}$ for each $o \in D$\;
    
    \textbf{for each}{$o \in D$} \textbf{do}{
        $N[o] \gets \KNNG[o] \cup \mathbf{R}[o]$\;
    }
    
    \ForEach{$o \in D$}{
        \ForEach{neighbor pair $(u,v)$ in $N[o]$}{
            Compute $d=\delta(u,v)$\;
            Insert $(v,d)$ into $\KNNG[u]$ if improved\;
            Insert $(u,d)$ into $\KNNG[v]$ if improved\;
        }
    }
}
\end{algorithm}

\sstitle{Index construction.}
Exact construction of ranked $\KNN$ graph requires $O(|D|^2)$ distance computations, since the $\mathbf{K}$-nearest neighbors must be identified for every point in $D$. 
To improve efficiency, NNDescent~\cite{NN-Descent-WWW-2011} iteratively refines the $\mathbf{K}$NN list of each point based on the neighbors of its current neighbors. 
Algorithm~\ref{algo:knng} shows how to create the ranked $\KNN$ graph using NNDescent.
It first initializes the neighbor list $\KNNG[o]$ of each point $o \in D$ with $\mathbf{K}$ random points (Line~1). During each iteration, it constructs a reverse-neighbor list $\mathbf{R}[o]$ for every point $o$, which contains all points that currently include $o$ in their neighbor lists (Line~3). 
The forward and reverse neighbor lists are then merged to form the neighbor set $N[o]$ (Line~4). 
For each neighbor pair $(u,v)$ in $N[o]$ (Lines~5--6), the algorithm computes their distance $\delta(u,v)$ and attempts to insert $(v,d)$ into $\KNNG[u]$ and $(u,d)$ into $\KNNG[v]$ if the new neighbor improves the current ranked neighbor list (Lines~7--9).
Here, we say that $(v,d)$ updates $\KNNG[u]$ if the distance $d=\delta(u,v)$ is smaller than the distance between $u$ and at least one neighbor currently stored in $\KNNG[u]$.
The refinement process repeats until no neighbor list changes further (Line~10).
For simplicity, we omit several optimizations, such as neighbor sampling and lazy updates.

\stitle{Navigation Graph.}
Ranked $\KNN$ graphs mainly connect each point to its local neighbors and therefore lack long-range edges. 
As a result, greedy graph traversal from a random entry point can easily become trapped in local optima~\cite{GreedySearch-2011-IJCAI}. To address this limitation, navigation graphs~\cite{NSG-VLDB-2019-Fu}, such as the Hierarchical Navigable Small World ($\HNSW$) graph~\cite{HNSW-2020-TPAMI}, have been proposed.

\begin{definition}[$\HNSW$ Graph]\label{def:hnsw}
Given a dataset $D \subset \mathbb{R}^d$, an $\HNSW$ graph is an $L$-layer directed proximity graph $\HNSWG =\{G^{(0)},G^{(1)},\ldots,G^{(L)}\}$, where each layer $G^{(\ell)}=(V^{(\ell)},E^{(\ell)})$ is defined over a subset $V^{(\ell)} \subseteq D$, satisfying $V^{(L)} \subseteq V^{(L-1)} \subseteq \cdots \subseteq V^{(0)} = D$.
Each point $o \in D$ is assigned a maximum level $\lambda(o)$ and appears in all layers $G^{(0)},\ldots,G^{(\lambda(o))}$. 
At each layer $\ell \leq \lambda(o)$, point $o$ maintains a bounded sized neighbor list $G^{(\ell)}[o]=\{v_1,\ldots,v_M\}$,
where each $v_i \in V^{(\ell)}$ is selected using proximity-based pruning heuristics. 
A directed edge $o \rightarrow v_i$ exists in layer $G^{(\ell)}$ if $v_i \in G^{(\ell)}[o]$.
\end{definition}

\begin{algorithm}[t]
\small
\caption{\textsc{Graph-Search}$(G,q,k,ep,ef)$}
\label{alg:graph-search}

\KwIn{Graph $G$, query vector $q$, target number of results $k$, entry point $ep$, beam width $ef$}
\KwOut{Candidate set $W[q]$ of size $k$}

Initialize min-heap $Q \gets \{ep\}$ ordered by $\delta(q,\cdot)$\;
Initialize max-heap $W[q] \gets \{ep\}$ ordered by $\delta(q,\cdot)$\;
Mark $ep$ as visited\;

\While{$Q \neq \emptyset$}{
    Extract the closest vertex $v$ from $Q$\;
    Let $u$ be the farthest vertex in $W[q]$\;
    
    \textbf{if}{$\delta(q,v) > \delta(q,u)$} \textbf{then}{
        \textbf{break}\;
    }
    
    \ForEach{unvisited neighbor $o \in G[v]$}{
        Mark $o$ as visited\;
        
        \If{$|W[q]| < k$ \textbf{or} $\delta(q,o) < \delta(q,u)$}{
            Insert $o$ into both $Q$ and $W[q]$\;
            
            \If{$|W[q]| > k$}{
                Remove the farthest vertex from $W[q]$\;
            }
            \If{$|Q| > ef$}{
                Remove the farthest vertex from $Q$\;
            }
        }
    }
}

\Return{$W[q]$}
\end{algorithm}

\vspace{-1em}

\sstitle{Search on the $\HNSW$ graph.}
Given a query $q$, an $\HNSW$ graph $\HNSWG$, the target result size $k$, and a beam width $ef$, search on the $\HNSW$ index proceeds from the top layer to the bottom layer. 
Starting from the global entry point in the top layer, the algorithm first performs greedy search on each upper layer with beam width $ef=1$, so that only the closest candidate is retained and used as the entry point for the next lower layer. 
After reaching the bottom layer $G^{(0)}$, the algorithm performs beam search with a larger beam width $ef>1$ to retrieve the final $k$NN result for the query vector $q$.

The search within each layer $G^{(\ell)}$ follows Algorithm~\ref{alg:graph-search}. 
Given the current entry point $ep$, the algorithm initializes a candidate heap $Q$ and a result heap $W[q]$ with $ep$ (Lines~1--2), and marks it as visited (Line~3). 
It then repeatedly extracts the closest vertex $v$ from $Q$ (Line~5) and compares it with the farthest vertex $u$ in $W[q]$ (Line~6).
If $v$ is farther from $q$ than $u$, the search stops because no remaining candidate can improve the current result (Line~7). Otherwise, the algorithm scans each unvisited neighbor $o \in G^{(\ell)}[v]$ (Line~8), marks it as visited (Line~9), and inserts it into both $Q$ and $W[q]$. 
if $W[q]$ has fewer than $k$ vertices or $o$ is closer to $q$ than the current farthest candidate (Lines~10--11). 
If $W[q]$ or $Q$ exceeds their size limits, the farthest vertex is removed (Line~12--15). 
The search returns $W[q]$ as the result (Line~16).

\ifthenelse{\isundefined{\themacro}}{}{
\begin{figure}[t]
    \centering
    \includegraphics[width=1\linewidth]{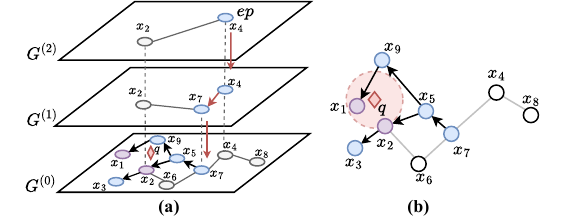}
     \vspace{-3em}
    \caption{Illustration of an $\HNSW$ graph and the graph search process.}
    \label{fig:hnsw-graph-search}
\end{figure}

\begin{example}
    Fig.~\ref{fig:hnsw-graph-search}(a) illustrates an $\HNSW$ graph with three layers, where $G^{(2)}$ is the top layer and $G^{(0)}$ is the bottom layer. Given a query vector $q$, the top-down routing phase starts from the global entry point $ep=x_4$ and performs greedy search ($ef=1$) on each layer. The search progressively descends through the hierarchy until reaching the bottom layer, where $x_7$ is identified as the entry point.
    Fig.~\ref{fig:hnsw-graph-search}(b) illustrates the bottom-layer beam search ($ef=2$). Starting from $x_7$, the search visits $x_5$ and inserts $x_2$ and $x_9$ into the candidate queue. Since $\delta(q,x_2)<\delta(q,x_9)$, $x_2$ is expanded first, followed by $x_9$, which discovers $x_1$. The final 2NN result is $W[q]={x_1,x_2}$ (shaded in purple).
\end{example}}

\vspace{-1em}
\section{Problem Analysis}\label{sec:problem-analysis}
We summarize existing methods for R$k$NN search in Section~\ref{sub:Existing}, and then analyze their two limitations in Section~\ref{sub:Limitation}. 
This analysis motivates our method presented in Section~\ref{sec:main-framework}.
For simplicity, we use the term R$k$NN search to refer to approximate R$k$NN search throughout the remainder of the paper unless otherwise specified.


\subsection{Existing Solutions}\label{sub:Existing}
The main challenge in answering the R$k$NN search problem lies in two aspects. First, unlike $k$NN search, where the result vectors are directly close to the query vector $q$, the vectors in the R$k$NN set do not necessarily have explicit proximity relationships with $q$. Specifically, R$k$NN search is defined from the perspective of dataset points $o \in D$: we need to determine whether $q$ belongs to the $k$NN set of points $o$ by checking the $k$NN-radius of $o$. Second, unlike $k$NN search, which always returns exactly $k$ vectors, the number of reverse $k$NN results for a query vector is unknown in advance.

To process the R$k$NN search problem for a query vector $q$ over dataset $D$, an intuitive approach is to directly follow the definition of R$k$NN. Specifically, the algorithm scans all points $o \in D$ and checks whether the distance $\delta(q,o)$ is no larger than the $k$NN-radius $r_k(o)$ of $o$, where $r_k(o)$ is defined as the distance from $o$ to its $k$-th nearest neighbor. 
If so, $o$ is included in the R$k$NN set of $q$.
However, this approach is computationally expensive for large-scale datasets.

\vspace{-0.3em}

\stitle{Filter-and-Verification Framework.}
To improve efficiency, existing methods typically adopt a filter-and-verification framework. Instead of examining all points in $D$, these methods first reduce the search space to a candidate set $C \subseteq D$, and then verify whether each candidate point belongs to the R$k$NN result set of the query vector $q$.
To obtain a compact candidate set $C$, existing methods commonly assume that \textbf{data points close to $q$ are more likely to regard $q$ as one of their $k$ nearest neighbors, and thus are more likely to belong to the R$k$NN result set of $q$.} Based on this assumption, methods such as $\SFT$~\cite{SFT-2003} and $\RDT$~\cite{RDT-2017-VLDB} first perform nearest-neighbor search and then examine nearby neighbors as candidates until a stopping condition is satisfied.

\input{figures/exp_non_locality}

\sstitle{The State-of-the-Art.}
More recently, $\HAMG$~\cite{HAMG-ICDE-2024} was proposed following the same intuition. A key insight of $\HAMG$ is that, on the Monotonic Relative Neighborhood ($\kw{MRN}$) graph~\cite{HNSW-2020-TPAMI, NSG-VLDB-2019-Fu, TauMG-2023-SIGMOD, SSSG-2022-TPAMI, DiskANN-NIPS-2019, Graph-index-survey-VLDB-2021-Wang, RoarGraph-VLDB-2024-Chen, NN-Descent-WWW-2011, Flash-Graph-Index-2025-SIGMOD, Graph-tree-index-survey-2023-IEEE-DataEngBull, AlphaCG-2026-SIGMOD-Lu, Graph-Theory-2023-Indyk, Fast-Graph-2025-SIGMOD-Lu, GreedySearch-2011-IJCAI, LSG-ICDE-2025-Wang, HVS-2021-VLDB, ELPIS-2023-VLDB, LMG-2025-ACMTDS-Xie}, the R$k$NN results of a query vector $q$ are guaranteed to lie within the $k$-hop neighborhood of $q$. Based on this property, $\HAMG$ treats these nearby graph neighbors as the candidate set $C$. Specifically, $\HAMG$ adopts the hierarchical structure of the ($\HNSW$) graph in Definition~\ref{def:hnsw}, while modifying the bottom layer $G^{(0)}$ to better approximate the $\kw{MRN}$ graph structure. Using this adapted ($\HNSW$) graph, $\HAMG$ retrieves the $k$-hop neighbors of $q$ as candidates and then verifies whether $\delta(q,o)\le r_k(o)$ for each candidate $o \in C$ to determine if $o$ belongs to the R$k$NN set of $q$.
More details can be found in~\cite{HAMG-ICDE-2024}.

\ifthenelse{\isundefined{\themacro}}{\vspace{-1.5em}}{}

\subsection{Limitation Analysis}\label{sub:Limitation}
Existing methods avoid scanning the entire dataset $D$ to retrieve the R$k$NN results of a query vector $q$. Yet, they still suffer from two major limitations that reduce their efficiency on large datasets.

\begin{figure}[t]
    \centering
    \captionsetup[sub]{skip=-0.4em} 
    \begin{small}
    \begin{tikzpicture}
    \begin{customlegend}[legend columns=2,
    legend entries={SIFT, MSMARCO},
    legend style={at={(0.5,1.15)},anchor=north,draw=none,font=\scriptsize,column sep=0.3cm}]
    \addlegendimage{area legend,fill=navy!70,draw=navy}
    \addlegendimage{area legend,fill=orange!70,draw=orange}
    \end{customlegend}
    \end{tikzpicture}
    \\[-\lineskip]
\subfloat[Total query latency (ms)]{
    \begin{tikzpicture}
    \begin{axis}[
        ybar, ymode=log, log origin=infty,
        width=0.48\columnwidth, height=2.5cm,
        bar width=3pt,
        symbolic x coords={SFT,RDT,HAMG,HRNN}, xtick=data,
        ymin=0.5, ymax=8000,
        ytick={1,10,100,1000},
        tick label style={font=\scriptsize}, xticklabel style={rotate=45,anchor=east,font=\fontsize{5}{5}\selectfont,yshift=1pt},
        ymajorgrids=true, grid style=dashed, enlarge x limits=0.20,
        nodes near coords, point meta=explicit symbolic,
        nodes near coords style={font=\fontsize{5}{5}\selectfont,rotate=90,anchor=west},
    ]
    \addplot+[fill=navy!70,draw=navy]   coordinates {(SFT,18.0)[18] (RDT,32.2)[32] (HAMG,119.5)[120] (HRNN,2.05)[2.1] };
    \addplot+[fill=orange!70,draw=orange] coordinates {(SFT,56.9)[57] (RDT,127)[127] (HAMG,471)[471] (HRNN,4.11)[4.1]};
    \end{axis}
    \end{tikzpicture}}\hspace{0.5cm}
\subfloat[Verification cost ($\mu$s/cand.)]{
    \begin{tikzpicture}
    \begin{axis}[
        ybar, ymode=log, log origin=infty,
        width=0.48\columnwidth, height=2.5cm,
        bar width=3pt,
        symbolic x coords={SFT,RDT,HAMG,HRNN}, xtick=data,
        ymin=0.05, ymax=40000,
        ytick={0.1,10,1000},
        tick label style={font=\scriptsize}, xticklabel style={rotate=45,anchor=east,font=\fontsize{5}{5}\selectfont,yshift=1pt},
        ymajorgrids=true, grid style=dashed, enlarge x limits=0.20,
        nodes near coords, point meta=explicit symbolic,
        nodes near coords style={font=\fontsize{5}{5}\selectfont,rotate=90,anchor=west},
    ]
    \addplot+[fill=navy!70,draw=navy]   coordinates {(SFT,322)[322] (RDT,347)[347] (HAMG,550)[550] (HRNN,0.14)[0.14]};
    \addplot+[fill=orange!70,draw=orange] coordinates {(SFT,1186)[1186] (RDT,1277)[1277] (HAMG,3542)[3542] (HRNN,0.38)[0.38] };
    \end{axis}
    \end{tikzpicture}}
    \end{small}
    \vspace{-0.4cm}
    \caption{Comparison of query performance for approximate R$k$NN search with $k=10$ at Recall@10 = 0.99. (a) shows the total query latency of $\HRNN$ and the baselines ($\SFT$, $\RDT$, and $\HAMG$). (b) reports the per-candidate verification cost of $\HRNN$ and the baselines.}
    \label{fig:decomp-cost}
\end{figure}
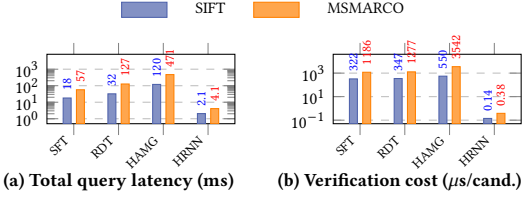

\vspace{-0.3em}
\stitle{Limitation 1: The mismatch between R$k$NN and nearby neighbors.}
Existing methods rely on the assumption that nearby neighbors of $q$ (e.g., nearest neighbors in $\SFT$ and $\RDT$, or graph neighbors within $k$ hops in $\HAMG$) contain all R$k$NN results. Although this assumption generally holds, it can still lead to high query latency in practice. As illustrated in the introduction (Fig.~\ref{fig:motivation}), when an object $o \in D$ has a large $k$NN-radius, the query vector $q$ may still belong to the $k$NN set of $o$ even if $o$ is far away from $q$. Consequently, the search algorithm must traverse a large number of nearby neighbors before reaching $o$. Even for the state-of-the-art $\HAMG$, the number of $k$-hop graph neighbors can become extremely large due to the graph expansion effect. As a result, a substantial amount of computation spent exploring nearby neighbors is unnecessary.

To further validate this limitation, we report the statistics in Fig.~\ref{fig:nonlocality-cdf}. We evaluate three datasets, GIST, SIFT, and MSMARCO, with dimensions 960, 128, and 1024, respectively, and set $k=10$ for R$k$NN search. For each query $q$, we compute its exact R$k$NN result set. For every object in the result set, we measure its rank in the nearest-neighbor ordering of $q$, where the closest neighbor has rank $1$ and the farthest point has rank $|D|$. We then compute the cumulative distribution function (CDF) of these rankings. The results show that all three datasets require exploring a large number of nearby neighbors to retrieve all R$k$NN results. For example, on SIFT, more than $50\%$ of queries require scanning over $10^3$ nearby neighbors to retrieve all R$k$NN results. More critically, some queries require exploring an extremely large neighborhood region. On GIST, certain queries need to scan more than $10^6$ nearby neighbors before all R$k$NN results can be identified. These observations indicate that relying solely on nearby neighbors is inefficient for R$k$NN search.

\stitle{Limitation 2: High Verification Cost.}
The first limitation implies that existing R$k$NN methods often explore a large number of unnecessary neighbors during the filter stage, producing a candidate set $C$ that is substantially larger than the true R$k$NN result set. 
Then, for each candidate vector $o \in C$, the algorithm must determine whether $\delta(o,q)\le r_k(o)$. If the condition does not hold, then $q$ is not among the $k$ nearest neighbors of $o$, and $o$ is discarded; otherwise, $o$ is reported as an R$k$NN result. Consequently, efficient computation of the $k$NN-radius $r_k(o)$ of $o$ is critical to R$k$NN search performance.
Obtaining this value requires issuing a $k$NN query centered at $o$, as precomputing $k$NN-radius and supporting queries for arbitrary $k$ values necessitate either prohibitive preprocessing cost or even complete index reconstruction in existing indexes \cite{SFT-2003, RDT-2017-VLDB, HAMG-ICDE-2024, Rdnn-tree-2001-ICDE}. 
Besides, since the candidate set $C$ is not known in advance and can be very large in practice, existing methods incur substantial overhead by computing $k$NN-radius online during query processing.

To quantify this limitation, we conduct experiments on the real-world SIFT and MSMARCO datasets. We compare the runtime of approximate R$k$NN search across $\HRNN$ and the baseline methods ($\SFT$, $\RDT$, and $\HAMG$), with the results reported in Fig.~\ref{fig:decomp-cost}. As shown in Fig.~\ref{fig:decomp-cost}(a), $\HRNN$ achieves up to one order of magnitude lower query latency than the baselines. Fig.~\ref{fig:decomp-cost}(b) further reveals the underlying reason: existing methods spend significantly more time verifying each candidate because they must compute the $k$NN-radius online. These results demonstrate that online $k$NN-radius computation is a major bottleneck for existing R$k$NN methods.

\vspace{-0.5em}
\section{The Proposed Method} \label{sec:main-framework}
We first introduce the proposed index structure in Section~\ref{sub:Structure}, followed by the query processing algorithm in Section~\ref{sub:Query}. 
Then, Section~\ref{sub:Construction} describes the index construction procedure, and Section~\ref{sub:maintenance} presents the append-only index maintenance method.

\vspace{-1em}

\subsection{Index Structure}\label{sub:Structure}
We now present the key ideas for addressing the two limitations of existing R$k$NN methods.

\stitle{Idea 1.}
To address Limitation 1, we revisit the relationship between $k$NN and R$k$NN search. 
Existing methods often assume that nearby vectors tend to share common properties, and therefore treat the nearest neighbors of the query as candidates for R$k$NN search. However, as shown in our analysis, nearby neighbors are often not true R$k$NN results. Instead, we make a different assumption: \textbf{a query's R$k$NN results can often be discovered through the R$k$NN results of its nearby neighbors.}
This intuition is empirically validated in Fig.~\ref{fig:witness-rank-cdf}.

\input{figures/proxy_knn_rank_at_gt}

This assumption follows from spatial proximity: as $q$ and its nearby neighbors lie close to each other, they are likely to fall within the $k$NN radius of the same point $x$. 
Thus, if $x$ is an R$k$NN of $q$'s nearby neighbors, it is also likely to appear in the R$k$NN sets of $q$.
Based on this intuition, we first retrieve a set of nearby vectors as \emph{proxies}. 
Rather than treating these proxies as R$k$NN candidates, we leverage their reverse-neighbor lists, which provide their R$k$NN results for arbitrary $k$, to generate candidates for the query. 
This design avoids excessive neighborhood expansion while keeping high recall.

\stitle{Idea 2.}
To address Limitation 2, we eliminate the expensive online computation of $k$NN-radius during verification. Recall that the $k$NN-radius of a vector $o \in D$ is simply the distance between $o$ and its $k$-th nearest neighbor. Thus, if we materialize a ranked neighbor list for each vector and store the corresponding distances in ascending order, the $k$NN-radius for any value of $k$ can be retrieved directly through a lookup. This transforms verification from an expensive online search into a lightweight index-access operation.

\stitle{Index Design.}
To realize the above intuitions, the index must support three operations: (1) efficiently retrieving nearby neighbors of a query vector as proxies, (2) obtaining the reverse-neighbor lists of proxy points to generate R$k$NN candidates for the query, and (3) maintaining ranked neighbor lists for all points in the dataset so that $k$NN-radius can be accessed efficiently during verification.

Designing a single index that satisfies all these requirements is non-trivial. For example, the $\HNSW$ graph is highly effective for nearest-neighbor retrieval but does not maintain ranked neighbor lists. In contrast, a ranked $\KNN$ graph provides ranked neighbor information but is inefficient for query-time nearest-neighbor search. Moreover, neither structure directly supports efficient access to the R$k$NN sets of data vectors.
To address these challenges, we propose $\HRNN$, a hybrid index that integrates an $\HNSW$ graph, a ranked $\KNN$ graph, and reverse-neighbor lists.

\begin{definition}[$\HRNN$ Index]\label{defi:hrnn}
Given a dataset $D$ and a user-defined parameter $\mathbf{K}$, the $\HRNN$ index is defined as $\mathcal{I}=(\HNSWG, \KNNG, \mathbf{R})$, where:

(1) $\HNSWG$ is an $\HNSW$ graph constructed over $D$ (Definition~\ref{def:hnsw});

(2) $\KNNG$ is a ranked $\KNN$ graph constructed over $D$ with parameter $\mathbf{K}$ (Definition~\ref{def:knng});

(3) $\mathbf{R}$ is the reverse-neighbor lists derived from $\KNNG$. For each vector $o \in D$, its reverse-neighbor list $\mathbf{R}[o]$ stores all pairs $(v,j)$ such that $o$ is the $j$-th nearest neighbor of $v$ in $\KNNG$ (Definition~\ref{defi:RR}).
\end{definition}

\begin{figure}[t]
    \centering
    \includegraphics[width=\linewidth]{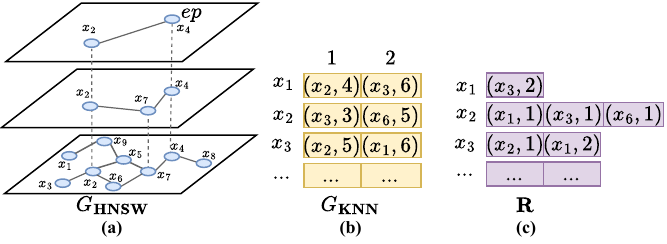}
    \vspace{-0.8cm}
    \caption{An example of the $\HRNN$ index structure.}
    \label{fig:hrnn-structure}\vspace{1em}
\end{figure}

\vspace{-0.8em}

\begin{example}
Fig.~\ref{fig:hrnn-structure} shows the components of $\HRNN$. Fig.~\ref{fig:hrnn-structure}(a) shows the $\HNSW$ graph. Fig.~\ref{fig:hrnn-structure}(b) presents the ranked $\KNN$ graph, where each node $o$ maintains a list of $\mathbf{K}$-nearest neighbors ($\mathbf{K}=2$), and each entry is denoted as $(v_i,\delta(o,v_i))$. For instance, $x_2$ is the 1NN of $x_3$ with $\delta(x_3,x_2)=5$. Fig.~\ref{fig:hrnn-structure}(c) shows the reverse-neighbor lists derived from $\KNNG$. For example, since $x_3$ is the 1NN of $x_2$ and the 2NN of $x_1$, its reverse-neighbor list is $\mathbf{R}(x_3)=\{(x_2,1),(x_1,2)\}$, which records these reverse-neighbor relationships.
\end{example}

\vspace{-0.8em}

\begin{theorem}
The $\HRNN$ index requires $O(|D|\cdot M \cdot L + |D|\cdot \mathbf{K})$ space, where $L$ is the number of $\HNSW$ layers, $M$ is the maximum degree of the $\HNSW$ graph, and $\mathbf{K}$ is the number of ranked neighbors materialized per point.
\end{theorem}

\ifthenelse{\isundefined{\themacro}}{}{
\begin{proof}
The $\HRNN$ index consists of three components. First, the $\HNSW$ graph stores at most $M$ outgoing neighbors for each vector at each layer. Since there are at most $|D|$ vectors and $L$ layers, the space required by the $\HNSW$ graph is bounded by $O(|D|\cdot M \cdot L)$. Second, the ranked $\KNN$ graph stores $\mathbf{K}$ ranked neighbors for each vector, requiring $O(|D|\cdot \mathbf{K})$ space. Third, the reverse-neighbor lists are obtained by transposing the ranked $\KNN$ graph. Each ranked neighbor entry generates exactly one reverse posting, yielding $|D|\cdot \mathbf{K}$ reverse postings in total and therefore requiring $O(|D|\cdot \mathbf{K})$ space. Summing the three components gives $O(|D|\cdot M \cdot L + 2|D|\cdot \mathbf{K}) = O(|D|\cdot M \cdot L + |D|\cdot \mathbf{K})$.
\end{proof}}

\vspace{-1em}

\subsection{Query Processing}\label{sub:Query}
We now describe how to answer an approximate R$k$NN query for a query vector $q$ using the proposed $\HRNN$ index
$\mathcal{I}=(\HNSWG, \KNNG, \mathbf{R})$.
If the query vector $q$ were part of the indexed dataset $D$, its reverse-neighbor list $\mathbf{R}[q]$ would directly reveal the points whose neighborhoods contain $q$. 
However, in practice, query vectors are typically not part of the indexed dataset.
To address this issue, $\HRNN$ first retrieves a set of nearby points as \emph{proxies} and uses their reverse-neighbor lists to approximate the unavailable reverse-neighborhood information of $q$.

\stitle{Filter.}
Let $N_m(q)=\{b_1,\ldots,b_m\}$ denote the approximate $m$-nearest neighbors of $q$ returned by the $\HNSW$ graph. 
Rather than treating these points as R$k$NN candidates, $\HRNN$ uses them to discover candidate points. 
To avoid scanning excessively large reverse-neighbor lists $\mathbf{R}[b]$ of the proxy point $b \in N_m(q)$, we introduce a rank threshold $\Theta$ and only consider vectors whose ranks do not exceed $\Theta$. 
This design is based on the intuition that if the rank of an object $v$ with respect to a proxy $b$ exceeds $\Theta$, then $b$---and thus the query---is unlikely to be within the top-$\Theta$ nearest neighbors of $v$. 
Hence, $v$ can be safely pruned, as it is unlikely to be an R$\Theta$NN of $q$.
Formally, the candidate set of query $q$ is defined as
\begin{equation}
C_{m,\Theta}(q)
=
\bigcup_{b \in N_m(q)}
\{v \mid (v,j)\in \mathbf{R}[b],\ j\le \Theta\}.
\end{equation}

A larger $\Theta$ increases candidate coverage and potentially improves recall, but incurs higher candidate-generation and verification costs. 
Note that $\mathbf{K}$ is the number of neighbors materialized in a ranked $\KNN$ graph, while $\Theta\leq\mathbf{K}$ is a query-time parameter controlling the number of entries scanned from each reverse-neighbor list.

\stitle{Verification.}
After candidate generation, $\HRNN$ verifies whether each candidate $o\in C_{m,\Theta}(q)$ belongs to the final R$k$NN result set of $q$. 
Existing methods typically compute the $k$NN-radius of each candidate online, which incurs substantial overhead. 
In contrast, $\HRNN$ leverages the pre-materialized ranked $\KNN$ graph $\KNNG$. 
Since the graph stores the top-$\mathbf{K}$ nearest neighbors of every vertex and $\mathbf{K}\geq k$, the $k$-th nearest neighbor of $o$ can be directly retrieved from $\KNNG[o]$. 
Let $v_k$ denote this neighbor. 
The corresponding distance $\hat r_k(o)=\delta(o,v_k)$ serves as an estimate of the $k$NN-radius of $o$. Using this pre-materialized radius, a candidate $o$ is accepted if $\delta(q,o)\leq \hat r_k(o)$.
This verification process reduces candidate validation to a single distance comparison.

\begin{algorithm}[t]
\caption{Approximate R$k$NN Search}\label{algo:search}
\small
\KwIn{Index $\mathcal{I}=(\HNSWG, \KNNG, \mathbf{R})$, query vector $q$, target parameter $k$, proxy size $m$, rank threshold $\Theta$, $\HNSW$ search parameters $ep$ and $ef_s$}
\KwOut{AR$k$NN result set $R$}

$R \gets \emptyset,\ C \gets \emptyset$\;

$N_m(q) \gets \textsc{Graph-Search}(\HNSWG, q, m, ep, ef_s)$\;

\For{$b \in N_m(q)$}{
    \For{$(v,j) \in \mathbf{R}(b)$ in ascending order of $j$}{
        \lIf{$j > \Theta$}{break}
        $C \gets C \cup \{v\}$\;
    }
}

\For{$o \in C$}{
    $v_k \gets \KNNG[o,k]$\;
    $\hat{r}_k(o) \gets \delta(o,v_k)$\;
    \textbf{if} {$\delta(q,o) \le \hat{r}_k(o)$} \textbf{then}{
        $R \gets R \cup \{o\}$\;
    }
}
\Return{$R$}
\end{algorithm}

\stitle{Algorithm.}
Algorithm~\ref{algo:search} presents the procedure for answering an approximate R$k$NN query for a query vector $q$ using the hybrid index $\mathcal{I}=(\HNSWG, \KNNG, \mathbf{R})$.
The algorithm first performs an approximate $m$-nearest-neighbor search on $\HNSWG$ to retrieve a set of proxy vertices $N_m(q)$ that are close to $q$ (Line~2). The details of the graph search procedure are provided in Algorithm~\ref{alg:graph-search}.
For each proxy point $b \in N_m(q)$, the algorithm scans its reverse-neighbor list $\mathbf{R}(b)$ in ascending rank order (Lines~3--4). Only vertices whose ranks do not exceed the threshold $\Theta$ are inserted into the candidate set $C$ (Lines~5--6). This proxy-guided reverse lookup generates candidate R$k$NN results without expensive neighborhood expansion.
The algorithm then verifies each candidate $o \in C$ (Line~7). Specifically, it retrieves the $k$-th nearest neighbor $v_k$ of $o$ from $\KNNG$ and obtains the corresponding estimated $k$NN-radius $\hat{r}_k(o)=\delta(o,v_k)$ (Lines~8--9). 
A candidate is accepted if $\delta(q,o)\leq \hat{r}_k(o)$ (Line~10).
All verified candidates are added to the result set $R$, which is returned as the final R$k$NN result (Line~11).

\begin{figure}
    \centering
    \includegraphics[width=1\linewidth]{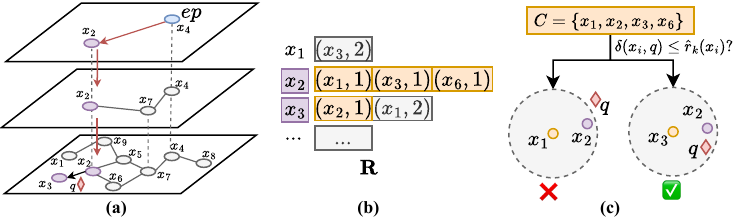}
    \vspace{-0.8cm}
    \caption{An example of the $\HRNN$ search process.}
    \label{fig:hrnn-search}
\end{figure}

\begin{example}
    Fig.~\ref{fig:hrnn-search} illustrates $\HRNN$ query processing for a query vector $q$, with $m=2$ and $\Theta=1$. 
    The search first traverses $\HNSWG$ and retrieves $x_2$ and $x_3$ as proxy vertices, as shown in Fig.~\ref{fig:hrnn-search}(a). 
    As illustrated in Fig.~\ref{fig:hrnn-search}(b), $\HRNN$ then scans the reverse-neighbor lists of the proxies in ascending rank order and inserts vertices whose rank does not exceed $\Theta$ into the candidate set $C$. 
    In this example, $x_1$, $x_2$, $x_3$, and $x_6$ are added to $C$. Finally, each candidate is verified by checking whether $q$ falls within its $k$NN-radius, which is obtained directly from $\KNNG$. 
    Candidates satisfying this condition (e.g., $x_3$) are returned as R$k$NN results.
\end{example}

\begin{theorem}
Given a query vector $q$, Algorithm~\ref{algo:search} answers the
R$k$NN query in expected time $O\!\left(
\log |D| + m + s(q) + d \cdot u(q)
\right)$, where $m$ is the proxy size, $s(q)$ is the number of scanned reverse-neighbor entries, and $u(q)$ is the number of distinct candidates after deduplication (i.e., the size of $C_{m,\Theta}(q)$ in Equation~1).
\end{theorem}

\ifthenelse{\isundefined{\themacro}}{}{
\begin{proof}
Proxy retrieval over the $\HNSW$ graph requires $O(\log |D| + m)$ expected time. Candidate generation scans $s(q)$ reverse-neighbor entries and therefore takes $O(s(q))$ time. 
After deduplication, the algorithm verifies $u(q)$ candidates.
Each verification consists of one exact distance computation and one constant-time radius lookup, requiring $O(d)$ time. Thus, verification costs $O(d \cdot u(q))$. Summing the three terms yields
the stated bound.
\end{proof}}

\subsection{Index Construction}\label{sub:Construction}
Constructing the $\HRNN$ index $\mathcal{I}=(\HNSWG, \KNNG, \mathbf{R})$ is challenging because it consists of three coupled structures: the $\HNSW$ graph, the ranked $\KNN$ graph, and the reverse-neighbor lists. A straightforward approach is to construct them independently, but this would incur redundant neighborhood exploration and repeated nearest-neighbor computations.

To avoid this overhead, $\HRNN$ adopts a unified construction pipeline. It first builds the $\HNSW$ graph and records the neighbors discovered during the bottom-layer searches for each inserted vector. Since these neighbors already approximate the local neighborhood of the vector, they provide high-quality seeds for initializing its ranked $\KNN$ list. $\HRNN$ then applies NNDescent to refine these initial neighborhoods and construct a high-quality ranked $\KNN$ graph. Compared with random initialization, this strategy improves convergence and reduces construction cost. Finally, the reverse-neighbor lists are derived by transposing the ranked $\KNN$ graph: if a vector $u$ is the $j$-th nearest neighbor of a vector $v$, the pair $(v,j)$ is inserted into the reverse-neighbor list of $u$. In this way, $\HRNN$ constructs all three index structures efficiently while maximizing the reuse of neighborhood information across different index views.

\begin{algorithm}[t]
\caption{$\HRNN$ Index Construction}\label{alg:build}
\small
\KwIn{
Dataset $D \subset \mathbb{R}^d$;
$\HNSW$ parameters $M$, $ef_c$, and $m_L$;
ranked KNN size $\mathbf{K}$;
NNDescent iterations $T$
}

\KwOut{
Index $\mathcal{I}=(\HNSWG, \KNNG, \mathbf{R})$
}

\tcp{Phase 1: Build the $\HNSW$ graph}
$(\HNSWG, W) \gets \textsc{Build-HNSW}(D, M, ef_c, m_L)$\;

\tcp{Phase 2: Construct the ranked KNN graph}
Initialize $\KNNG[o]$ using $W[o]$ for each $o \in D$\;

Execute Lines~2--11 of Algorithm~\ref{algo:knng} to build $\KNNG$\;

\tcp{Phase 3: Construct reverse-neighbor lists}
Initialize $\mathbf{R}(o)\gets\emptyset$ for all $o\in D$\;

\ForEach{$o\in D$}{
    \For{$j\gets 1$ \KwTo $\mathbf{K}$}{
        $v\gets \KNNG[o,j]$\;
        Append $(o,j)$ to $\mathbf{R}(v)$\;
    }
}

\ForEach{$o\in D$}{
    Sort $\mathbf{R}(o)$ in ascending order of rank $j$\;
}

\Return{$\mathcal{I}=(\HNSWG, \KNNG, \mathbf{R})$}
\end{algorithm}

\stitle{Algorithm.}
Algorithm~\ref{alg:build} constructs the $\HRNN$ index
$\mathcal{I}=(\HNSWG,\KNNG,\mathbf{R})$ in three phases.
First, $\HRNN$ builds the $\HNSW$ graph $\HNSWG$ over the dataset $D$ (Line~1). 
During $\HNSW$ construction, each point is inserted incrementally. For a newly inserted point $o$, $\HNSW$ first performs graph search from the current entry point to retrieve a set of approximate nearest neighbors using Algorithm~\ref{alg:graph-search}. 
These searched neighbors are then used to connect $o$ to existing vertices in different $\HNSW$ layers. 
Meanwhile, $\HRNN$ records the bottom-layer search results as $W[o]$, which provides high-quality initial neighbor candidates for constructing the ranked $\KNN$ graph.
Second, $\HRNN$ constructs the ranked $\KNN$ graph $\KNNG$ (Lines~2--3). 
Instead of randomly initializing the neighbor list of each vertex, $\HRNN$ initializes $\KNNG[o]$ using $W[o]$, i.e., the candidate neighbors obtained during $\HNSW$ construction. 
Then, Lines~2--11 of Algorithm~\ref{algo:knng} are executed to refine these initial lists via NNDescent to build a $\KNNG$.
Third, $\HRNN$ builds the reverse-neighbor lists $\mathbf{R}$ from $\KNNG$ (Lines~4). 
For each vertex $o$ and each rank position $j$ (Line~5--6), the algorithm obtains the $j$-th nearest neighbor $v=\KNNG[o,j]$ (Line~7) and appends $(o,j)$ to $\mathbf{R}(v)$ (Line~8).
Finally, each reverse-neighbor list is sorted by rank $j$ (Line~9--10).

\ifthenelse{\isundefined{\themacro}}{}{
\begin{figure}
    \centering
    \includegraphics[width=1\linewidth]{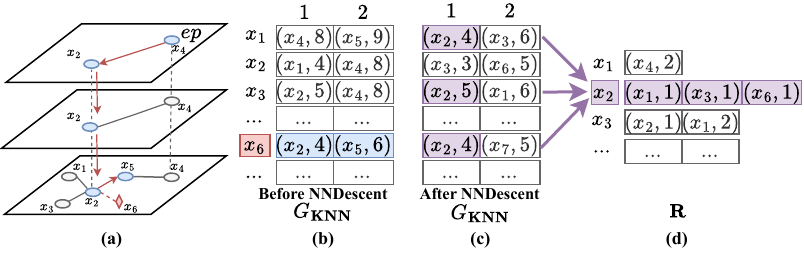}
    \vspace{-0.8cm}
    \caption{An example of $\HRNN$ index construction.}
    \label{fig:hrnn-indexing}
\end{figure}

\begin{example}
    Fig.~\ref{fig:hrnn-indexing} illustrates the construction process of the $\HRNN$ index. 
    Fig.~\ref{fig:hrnn-indexing}(a) shows the  search performed when inserting a new point $x_6$ into $\HNSWG$, and $x_6$ is connected to $x_2$ after pruning the search results. 
    Other points are inserted similarly. 
    The search result $[x_2, x_3]$ is recorded and used to initialize the ranked $\KNN$ graph with $\mathbf{K}=2$, as shown in Fig.~\ref{fig:hrnn-indexing}(b). 
    Fig.~\ref{fig:hrnn-indexing}(c) shows how NNDescent refines the ranked $\KNN$ graph. 
    For example, in $x_6$'s neighbor list, the initial 2NN $x_3$ is replaced by the closer neighbor $x_7$. 
    Finally, Fig.~\ref{fig:hrnn-indexing}(d) constructs the reverse-neighbor lists by scanning the refined ranked $\KNN$ graph. For instance, since $x_2$ appears as the 1NN of $x_1$, $x_3$, and $x_6$, its reverse-neighbor list is $\mathbf{R}(x_2)=[(x_1,1),(x_3,1),(x_6,1)]$.
\end{example}}

\vspace{-0.5em}
\begin{theorem}
Algorithm~\ref{alg:build} constructs the $\HRNN$ index in expected time $O\left(
|D|  \cdot \log |D|
+
T \cdot |D| \cdot \mathbf{K}^{2} \cdot d
+
|D| \cdot \mathbf{K}
\right)$,
where $d$ is the vector dimension, 
$\mathbf{K}$ is the ranked neighbor-list size, and $T$ is the number of NNDescent refinement iterations.
\end{theorem}

\ifthenelse{\isundefined{\themacro}}{}{
\begin{proof}
Algorithm~\ref{alg:build} consists of three phases.
First, the algorithm builds the $\HNSW$ graph over $D$.
$\HNSW$ inserts vectors incrementally. For each inserted vector, it performs graph search to locate candidate neighbors and then selects at most $M$ neighbors for graph construction. Following
the standard logarithmic search behavior of $\HNSW$, the expected cost of inserting one vector is $O(\log |D|)$. Therefore, inserting all $|D|$ vectors costs $O(|D|  \cdot \log |D|)$.
During this phase, $\HNSW$ also records the bottom-layer search results as initialization seeds for the ranked $\KNN$ graph. This recording is performed during $\HNSW$ construction and does not change the asymptotic cost.

Second, $\HRNN$ constructs the ranked $\KNN$ graph. Each
vector maintains a ranked list of $\mathbf{K}$ neighbors. The initial
lists are obtained from the $\HNSW$ search results and are refined using NNDescent. In each refinement iteration, for each vector, NNDescent compares candidate pairs generated from its forward and reverse neighbor lists. Since each list has size at most $\mathbf{K}$, the number of candidate comparisons per vector is $O(\mathbf{K}^{2})$. Each comparison requires computing the distance between two $d$-dimensional vectors, which costs $O(d)$. Hence,
one NNDescent iteration costs $O(|D| \cdot \mathbf{K}^{2} \cdot d)$.
After $T$ iterations, the total refinement cost is $O(T \cdot |D| \cdot \mathbf{K}^{2} \cdot d)$.
Third, $\HRNN$ constructs the reverse-neighbor lists by transposing the ranked $\KNN$ graph. Since each of the $|D|$ vectors stores $\mathbf{K}$ ranked neighbors, the ranked $\KNN$ graph contains exactly $|D| \cdot \mathbf{K}$ directed neighbor entries. For each entry, $\HRNN$ inserts one corresponding reverse posting. Thus, materializing the reverse-neighbor lists costs $O(|D| \cdot \mathbf{K})$.
Summing the costs of the three phases gives the result.
\end{proof}}

\vspace{-1em}

\subsection{Index Maintenance}\label{sub:maintenance}
Real-world vector datasets are rarely static. New vectors are continuously generated in applications such as retrieval-augmented generation and recommendation systems. Rebuilding the entire $\HRNN$ index after each insertion is prohibitively expensive.
Given a newly inserted point $o_{\text{new}}$, maintenance must synchronize three coupled structures. First, $o_{\text{new}}$ must be incorporated into the $\HNSW$ graph to support future proxy retrieval. Second, if $o_{\text{new}}$ becomes one of the top-$\mathbf{K}$ nearest neighbors of existing vectors, the corresponding ranked neighbor lists in $\KNNG$ must be updated. Third, all affected reverse-neighbor lists must be refreshed to remain consistent with the updated ranked KNN graph.
Thus, we first insert $o_{\text{new}}$ into $\HNSWG$ and then describe how to minimize the maintenance cost of updating $\KNNG$ and $\RR$.

\stitle{Exact Affected Area.}
The key challenge is identifying the existing points in dataset $D$ affected by the insertion of $o_{\text{new}}$. 
Let $r_{\mathbf{K}}(o)$ denote the distance from $o$ to its current $\mathbf{K}$-th nearest neighbor in $\KNNG$. 
A vector $x$ is affected if $o_{\text{new}}$ should appear in its top-$\mathbf{K}$ neighbor list after insertion. 
Formally, the exact affected set is
\begin{equation}
S^\star(o_{\text{new}})
=
\{x \in D \mid
\delta(x,o_{\text{new}}) < r_{\mathbf{K}}(x)\}.
\end{equation}

For every point $x \in S^\star(o_{\text{new}})$, inserting $o_{\text{new}}$ modifies the ranked neighbor list of $x$, which subsequently triggers updates to reverse-neighbor lists. However, computing $S^\star(o_{\text{new}})$ exactly requires evaluating $\delta(x,o_{\text{new}})$ for every point in the dataset, resulting in an $O(|D|)$ scan per insertion. Such a cost is impractical for large-scale vector collections.

\stitle{Approximate Affected Area.}
To avoid exhaustive scans, $\HRNN$ adopts the same proxy principle used during query processing. Intuitively, points close to $o_{\text{new}}$ tend to share similar reverse-neighbor relationships and therefore provide effective hints about the affected region.
During $\HNSW$ insertion, $\HRNN$ obtains a set of nearby proxy points $N_{m_u}(o_{\text{new}})$ and records the search result $W(o_{\text{new}})$ used to initialize the ranked neighbor list of $o_{\text{new}}$. Given an update-time rank threshold $\Theta_u \le \mathbf{K}$, $\HRNN$ constructs an approximate affected set by scanning the reverse-neighbor lists of these proxies:
\begin{equation}
\hat S_u(o_{\text{new}})
=
\bigcup_{b \in N_{m_u}(o_{\text{new}})}
\{x \mid (x,j)\in \mathbf{R}(b),\ j \le \Theta_u\}.
\end{equation}

The parameter $\Theta_u$ controls the trade-off between maintenance cost and update quality. Larger values increase the coverage of the affected region but require scanning more reverse-neighbor entries.

\stitle{Maintenance Procedure.}
After identifying the approximate affected set, $\HRNN$ incrementally updates all index components.
Specifically, $o_{\text{new}}$ is first inserted into the $\HNSW$ graph, and the collected search result $W(o_{\text{new}})$ is used to initialize its ranked neighbor list in $\KNNG$. 
Next, $\HRNN$ examines each point $x \in \hat S_u(o_{\text{new}})$. If $\delta(x,o_{\text{new}})$ is smaller than the distance to the current $\mathbf{K}$-th neighbor of $x$, the new vector is inserted into $\KNNG[x]$. If the list already contains $\mathbf{K}$ entries, the current farthest neighbor is evicted. Whenever a neighbor list changes, the corresponding reverse-neighbor lists are updated by removing obsolete entries, inserting new entries, and adjusting the ranks of shifted neighbors.
Since the affected area is identified approximately rather than exactly, the maintained index may differ slightly from one obtained through full batch reconstruction. However, as demonstrated in Section~\ref{sec:experiment}, this discrepancy is negligible in practice. Periodic reconstruction can be performed when necessary to restore the index to its batch-built state.
\ifthenelse{\isundefined{\themacro}}{The complete algorithm is provided in our technical report~\cite{technicalreport}.}{}

\ifthenelse{\isundefined{\themacro}}{}{
\begin{algorithm}[!htb]
\caption{HRNN Index Insertion Maintenance}
\label{alg:insert-maintenance}
\small
\KwIn{Index $\mathcal{I}=(\HNSWG,\KNNG,\mathbf{R})$, new vector $o_{\mathrm{new}}$, update parameters $m_u, \Theta_u$}
\KwOut{Updated index $\mathcal{I}$}

\tcp{Phase 1: Insert into $\HNSW$ and collect proxy vectors}
$(\HNSWG,W(o_{\mathrm{new}}))
\gets
\textsc{HNSW-Insert}(\HNSWG,o_{\mathrm{new}})$\;

$B_{m_u}(o_{\mathrm{new}})
\gets
$ top-$m_u$ closest vectors in $W(o_{\mathrm{new}})$ to $o_{\mathrm{new}}$\;

\tcp{Phase 2: Construct the approximate affected area}
$\hat S_u(o_{\mathrm{new}})\gets\emptyset$\;

\ForEach{$b\in B_{m_u}(o_{\mathrm{new}})$}{
    \ForEach{$(x,j)\in \mathbf{R}(b)$ in ascending order of $j$}{
        \lIf{$j> \Theta_u$}{break}
        $\hat S_u(o_{\mathrm{new}})
        \gets
        \hat S_u(o_{\mathrm{new}})\cup\{x\}$\;
    }
}

\tcp{Phase 3: Initialize the new vector}
$L_{\mathrm{new}}
\gets
$ top-$\mathbf{K}$ vectors in $W(o_{\mathrm{new}})$
ranked by $\delta(o_{\mathrm{new}},\cdot)$\;

$\KNNG[o_{\mathrm{new}}]
\gets
L_{\mathrm{new}}$\;

\ForEach{$v\in L_{\mathrm{new}}$}{
    $j\gets$ rank of $v$ in $L_{\mathrm{new}}$\;
    Insert $(o_{\mathrm{new}},j)$ into $\mathbf{R}(v)$,
    keeping $\mathbf{R}(v)$ sorted by rank\;
}

\tcp{Phase 4: Refresh affected neighborhoods}
\ForEach{$x\in\hat S_u(o_{\mathrm{new}})$}{

    $L_{\mathrm{old}}
    \gets
    \KNNG[x]$\;

    $r_{\mathbf{K}}(x)
    \gets
    \delta(x,\KNNG[x,\mathbf{K}])$\;

    \If{$\delta(x,o_{\mathrm{new}})<r_{\mathbf{K}}(x)$}{

        $L_{\mathrm{upd}}
        \gets
        \textsc{TopK}(L_{\mathrm{old}}\cup\{o_{\mathrm{new}}\})$
        ranked by $\delta(x,\cdot)$\;

        $\KNNG[x]
        \gets
        L_{\mathrm{upd}}$\;

        \ForEach{$v\in L_{\mathrm{old}}\setminus L_{\mathrm{upd}}$}{
            Remove $(x,\cdot)$ from $\mathbf{R}(v)$\;
        }

        \ForEach{$v\in L_{\mathrm{upd}}$}{
            $j\gets$ rank of $v$ in $L_{\mathrm{upd}}$\;
            Insert or update $(x,j)$ in $\mathbf{R}(v)$,
            keeping $\mathbf{R}(v)$ sorted by rank\;
        }
    }
}

\Return{$\mathcal{I}$}
\end{algorithm}}

\ifthenelse{\isundefined{\themacro}}{}{
\stitle{Algorithm.}
Algorithm~\ref{alg:insert-maintenance} incrementally maintains the $\HRNN$ index after the arrival of a new vector $o_{\mathrm{new}}$. The procedure updates the three coupled structures of $\HRNN$: the $\HNSW$ graph $\HNSWG$, the ranked $\KNN$ graph $\KNNG$, and the reverse-neighbor lists $\mathbf{R}$.
In Phase~1, the algorithm inserts $o_{\mathrm{new}}$ into $\HNSWG$ and reuses the search result $W(o_{\mathrm{new}})$ produced during $\HNSW$ insertion (Lines~1--2). The top-$m_u$ closest vectors in this set are selected as proxy vectors $N_{m_u}(o_{\mathrm{new}})$, which serve as proxies for identifying the affected region.
In Phase~2, the algorithm constructs the approximate affected area $\hat S_u(o_{\mathrm{new}})$ (Lines~3--8). For each proxy vector $b$, it scans the reverse-neighbor list $\mathbf{R}(b)$ in ascending rank order and collects the owner vectors whose ranks do not exceed the threshold $\Theta_u$. This step follows the same proxy principle used during query processing and avoids the expensive computation of the exact affected area.
In Phase~3, the algorithm initializes the ranked neighbor list of $o_{\mathrm{new}}$ using the top-$\mathbf{K}$ vectors from $W(o_{\mathrm{new}})$ (Lines~9--14). The resulting neighborhood is inserted into $\KNNG$, and the corresponding reverse-neighbor entries are added to $\mathbf{R}$.
Finally, in Phase~4, the algorithm refreshes the neighborhoods of all vectors in $\hat S_u(o_{\mathrm{new}})$ (Lines~15--30). For each vector $x$, it first retrieves the current $\mathbf{K}$NN-radius $r_{\mathbf{K}}(x)$. If $\delta(x,o_{\mathrm{new}})<r_{\mathbf{K}}(x)$, the new vector belongs to the top-$\mathbf{K}$ neighborhood of $x$. The algorithm therefore inserts $o_{\mathrm{new}}$ into $\KNNG[x]$, removes the farthest neighbor if necessary, and constructs the updated ranked neighbor list. It then synchronizes the reverse-neighbor lists by removing obsolete postings and inserting or updating the postings associated with the modified neighborhood. Consequently, the ranked $\KNN$ graph and reverse-neighbor lists remain consistent after each insertion.
}

\begin{theorem}
For an insertion of $o_{\text{new}}$, the maintenance cost of
$\HRNN$ is
$O\!\left(
 \log |D|
+
s_u
+
d \cdot |\hat S_u(o_{\text{new}})|
+
(a_u+1)\cdot \mathbf{K}
\right)$, where 
$s_u$ is the number of scanned reverse-neighbor entries,
$\hat S_u(o_{\text{new}})$ is the approximate affected set, and
$a_u$ is the number of existing ranked neighbor lists updated.
\end{theorem}

\vspace{-0.2cm}

\ifthenelse{\isundefined{\themacro}}{}{
\begin{proof}
The maintenance procedure has four parts. First, $\HRNN$
inserts $o_{\text{new}}$ into the $\HNSW$ graph, and one insertion costs $O(\log |D|)$.
Second, $\HRNN$ discovers the approximate affected set by
scanning the rank-truncated reverse-neighbor lists of the update
proxies. Let $s_u$ denote the number of scanned reverse-neighbor
entries. This step costs $O(s_u)$.
Third, for each vector $x\in \hat S_u(o_{\text{new}})$, $\HRNN$
checks whether $o_{\text{new}}$ enters the top-$\mathbf{K}$ neighbor list
of $x$. Each check requires one exact distance computation
$\delta(x,o_{\text{new}})$ over $d$ dimensions and one comparison with
the current $\mathbf{K}$-th neighbor distance. Thus, this step costs $O(d \cdot |\hat S_u(o_{\text{new}})|)$.
Finally, suppose $a_u$ existing ranked neighbor lists are updated. Refreshing one changed list may shift at most $\mathbf{K}$ ranked positions and therefore touches $O(\mathbf{K})$ reverse-neighbor entries.
In addition, the ranked neighbor list of $o_{\text{new}}$ is initialized and materialized, which costs $O(\mathbf{K})$. Hence, the synchronization
cost is $O((a_u+1)\cdot \mathbf{K})$.
Summing the four terms gives the total cost.
\end{proof}}

\stitle{Remark.}
Building an $\HRNN$ index solely through incremental insertions is feasible, but incurs the cumulative cost of maintaining the index after each update.
Unlike batch construction in Algorithm~\ref{alg:build}, which amortizes neighbor refinement through a one-time NNDescent procedure, insertion-based construction performs affected-area discovery and reverse-neighbor-list synchronization for each arriving point. 
Consequently, its total cost depends on the distribution of update proxies, the sizes of the affected areas, and the number of modified reverse-neighbor lists. As shown in our experiments, insertion-only construction is slower than batch construction on static datasets. However, it avoids expensive index rebuilding and supports continuous updates while preserving the same index structure and query processing algorithm.

\section{Experiment} \label{sec:experiment}

\subsection{Experimental Setup}

\stitle{Datasets.}
We evaluate on four real-world vector datasets: SIFT~\cite{SIFT-IVFADC-2011-ICASSP}, Msong~\cite{MSONG-2012-WWW}, GIST~\cite{GIST-PQ-2011-TPAMI}, and MSMARCO~\cite{MSMARCO-Dataset-2016-CocoNIPS}. SIFT, Msong, and GIST are standard benchmarks widely used in prior approximate R$k$NN studies~\cite{HAMG-ICDE-2024}, while MSMARCO represents modern text-embedding retrieval workloads. 
Table~\ref{tbl:dataset} summarizes the dataset statistics.
Due to the $O(|D|^2)$ cost of generating exact R$k$NN ground truth, following prior work~\cite{HAMG-ICDE-2024}, our main experiments use datasets with 1 million vectors. 
To evaluate scalability, we further evaluate $\HRNN$ on datasets containing up to 10 million vectors in Exp-8.

\begin{table}[t]
\centering
\caption{Dataset statistics.}\vgap\vspace{-0.5em}
\label{tbl:dataset}
\small
\begin{tabular}{l|c|c|c|c}
\toprule
\textbf{Dataset} & \textbf{Dimension} & \textbf{\# Base} & \textbf{\# Query} & \textbf{Domain} \\
\midrule
SIFT     & 128  & 1.0M & 10,000 & Images \\
Msong    & 420  & 1.0M & 1,000  & Audio \\
GIST     & 960  & 1.0M & 1,000  & Images \\
MSMARCO  & 1024 & 1.0M & 1,000  & Text Embeddings \\
\bottomrule
\end{tabular}\vspace{1em}
\end{table}

\stitle{Compared Methods.}
We compare $\HRNN$ against representative graph-based approximate R$k$NN methods, including $\SFT$, $\RDT$, and the state-of-the-art $\HAMG$~\cite{HAMG-ICDE-2024}. Following $\HAMG$, we implement $\SFT$ and $\RDT$ on top of $\HNSW$ to provide strong high-dimensional baselines. We do not compare against spatial-partitioning methods~\cite{TPL-2004-VLDB,TPL2-2007-Springer,FINCH-2008-VLDB} or tree-based R$k$NN methods~\cite{ERkNN-2005-CIKM,MRKNNCop-2006-SIGMOD}, as prior studies have shown them to be prohibitively expensive on million-scale high-dimensional datasets.

\begin{itemize}[leftmargin=*]
\item \textbf{$\HAMG$.} The state-of-the-art AR$k$NN method~\cite{HAMG-ICDE-2024}. Since the original Java implementation is not publicly available and the paper does not fully specify all dataset-specific parameter settings, we re-implement $\HAMG$ in C++ based on the published algorithm description. We tune both construction and query parameters over the ranges reported in~\cite{HAMG-ICDE-2024}, as well as larger budgets, and report the best recall--latency trade-off achieved on our platform.

\item \textbf{$\HNSW$-$\SFT$.} The $\SFT$-style baseline used in $\HAMG$~\cite{HAMG-ICDE-2024,SFT-2003}. It first retrieves an expanded set of $k'$ nearest neighbors of the query using $\HNSW$ and then verifies each candidate for R$k$NN membership.

\item \textbf{$\HNSW$-$\RDT$.} The $\RDT$ baseline~\cite{RDT-2017-VLDB,HAMG-ICDE-2024} implemented on top of $\HNSW$, following the design adopted in $\HAMG$. It incrementally explores unchecked nearest neighbors around the query and applies the pruning and termination rules of $\RDT$.
\end{itemize}

\stitle{Metrics.}
We evaluate both efficiency and accuracy.
Query performance is measured by average queries per second (QPS) over all queries, while accuracy is measured by Recall@$k$. For index construction, we report construction time and index size. A construction process is terminated if its runtime exceeds $10^4$ seconds.

\stitle{Implementation.}
We implement all methods in C++ and compile them using GCC 13.1.0. Experiments are conducted on a server equipped with an Intel Xeon Platinum 8352V CPU @ 2.10GHz and 512\,GB RAM running Ubuntu 20.04 LTS. Following common practice in vector search studies~\cite{NSG-VLDB-2019-Fu,DiskANN-NIPS-2019,HNSW-2020-TPAMI,Graph-index-survey-VLDB-2021-Wang,HAMG-ICDE-2024}, all query experiments are executed using a single thread, while index construction uses 64 threads.
For $\HRNN$, the $\HNSW$ graph is constructed with $M=16$ and $ef_c=400$. For the baselines, we tune parameters according to the ranges reported in their original papers and use $C=8000$ and $d_m=80$ for $\HAMG$, and $ef_c=400$ for $\HNSW$-$\RDT$ and $\HNSW$-$\SFT$.

\stitle{Parameter Settings.}
Unless otherwise specified, all experiments evaluate AR$k$NN search with $k=10$. For $\HRNN$, we construct a single index per dataset with $\mathbf{K}=500$.
$\HRNN$ has two query-time parameters: the number of proxies $m$ and the reverse-neighbor threshold $\Theta$. 
To generate the recall--throughput trade-off curves in Exp-1, we evaluate $m \in \{1,3,5,10,20,30,50,100,150,200\}$ and $\Theta \in \{10,20,30,50,100,200,300,500\}$. 
Each reported $\HRNN$ point corresponds to a specific $(m,\Theta)$ configuration. For experiments requiring a single configuration at a target recall threshold $\tau$, we select the configuration that achieves the highest throughput among all evaluated settings satisfying Recall@$k \geq \tau$.

\vspace{-0.5em}
\subsection{Experimental Results}
\vspace{-0.5em}

\input{figures/main_result_grid_search}

\stitle{Exp-1: Query Performance Comparison.}
We first compare the end-to-end recall--throughput trade-off of all methods. Fig.~\ref{fig:main-experiment} reports the results of $\HRNN$, $\HAMG$, $\HNSW$-$\SFT$, and $\HNSW$-$\RDT$ on the four datasets.
On MSMARCO, at Recall@10 above 0.99, $\HRNN$ achieves 185 QPS, outperforming $\HAMG$, $\HNSW$-$\SFT$, and $\HNSW$-$\RDT$ by $83.1\times$, $10.3\times$, and $23.3\times$, respectively. 
On GIST, $\HAMG$ achieves only 0.714 recall. 
Although $\HNSW$-$\SFT$ and $\HNSW$-$\RDT$ can further improve recall by expanding the query-centered search region, their throughput drops to roughly 1 QPS or below. 
Specifically, $\HNSW$-$\SFT$ reaches 0.958 recall at only 0.372 QPS, while $\HNSW$-$\RDT$ reaches 0.880 recall at 1.076 QPS. 
We therefore terminate further expansion once the throughput becomes prohibitively low. In contrast, $\HRNN$ achieves 0.999 recall at 39.3 QPS.

These results expose the limitations of existing methods for high-dimensional R$k$NN search, especially on datasets with a pronounced mismatch between nearby neighbors and true R$k$NN results. By using nearby neighbors as proxies instead of direct candidates and verifying only reverse-lookup candidates with pre-materialized $k$NN-radius information, $\HRNN$ substantially reduces both candidate-generation and verification overhead, yielding consistently better recall--throughput trade-offs.

\stitle{Exp-2: Query Time Breakdown.}
We analyze how query time is distributed across different stages of $\HRNN$. 
We decompose each query into three phases: (1) proxy retrieval using the navigation graph $\HNSWG$, (2) reverse-neighbor list scanning for candidate generation, and (3) candidate verification using pre-materialized $k$NN-radius retrieved from the ranked $k$NN graph $\KNNG$. 
For each dataset, we select the lowest-latency configuration from the same $(m,\Theta)$ grid used in Exp-1 under Recall@10 thresholds of 0.95 and 0.99.

Fig.~\ref{fig:phase-breakdown} shows that proxy retrieval incurs a small and nearly constant cost across datasets and recall targets, ranging from 0.43ms to 1.7ms. 
Reverse-neighbor list scanning also remains lightweight because $\HRNN$ accesses only rank-truncated posting entries. In contrast, candidate verification becomes the dominant cost at high recall levels, as more candidates must be validated. For example, verification accounts for 78.5\%--86.3\% of the total query time on GIST and increases from 35.7\% to 74.3\% on MSMARCO as the target recall rises from 0.95 to 0.99.
Nevertheless, compared with existing methods (see also Fig.~\ref{fig:decomp-cost}), $\HRNN$ reduces verification overhead by replacing expensive online $k$NN-radius computation with direct lookups from the materialized ranked $k$NN graph, thereby mitigating one of the major bottlenecks of prior R$k$NN approaches.

\begin{figure}[t]
    \centering
    \captionsetup[sub]{skip=-2pt} 
    \begin{small}
    \begin{tikzpicture}
    \begin{customlegend}[legend columns=3,
    legend entries={Phase 1 Proxy Retrieval, Phase 2 $\mathbf{R}$ List Scan, Phase 3 Verify},
    legend style={at={(0.5,1.15)},anchor=north,draw=none,font=\scriptsize,column sep=0.1cm}]
    \addlegendimage{area legend,fill=navy!70,draw=navy}
    \addlegendimage{area legend,fill=amaranth!50,draw=amaranth}
    \addlegendimage{area legend,fill=orange!70,draw=orange}
    \end{customlegend}
    \end{tikzpicture}
    \\[-0.8em]
\subfloat[SIFT]{
    \begin{tikzpicture}
    \begin{axis}[
        ybar stacked,
        height=\columnwidth/3,
        width=0.235\textwidth,
        bar width=9pt,
        symbolic x coords={0.95,0.99},
        xtick=data,
        xlabel=Target Recall@10, xlabel style={font=\scriptsize, yshift=4pt},
        ylabel=Latency (ms), ylabel style={yshift=-4pt},
        ymin=0, ymax=1.27,
        label style={font=\scriptsize},
        tick label style={font=\scriptsize},
        ymajorgrids=true, grid style=dashed,
        enlarge x limits=0.5,
        legend style={draw=none},
    ]
    \addplot+[fill=navy!70,draw=navy] coordinates {(0.95,0.425) (0.99,0.429)};
    \addplot+[fill=amaranth!50,draw=amaranth] coordinates {(0.95,0.051) (0.99,0.115)};
    \addplot+[fill=orange!70,draw=orange] coordinates {(0.95,0.263) (0.99,0.559)};
    \end{axis}
    \end{tikzpicture}}
\subfloat[Msong]{
    \begin{tikzpicture}
    \begin{axis}[
        ybar stacked,
        height=\columnwidth/3,
        width=0.235\textwidth,
        bar width=9pt,
        symbolic x coords={0.95,0.99},
        xtick=data,
        xlabel=Target Recall@10,xlabel style={font=\scriptsize, yshift=4pt},
        ylabel=Latency (ms),ylabel style={yshift=-4pt},
        ymin=0, ymax=1.53,
        label style={font=\scriptsize},
        tick label style={font=\scriptsize},
        ymajorgrids=true, grid style=dashed,
        enlarge x limits=0.5,
        legend style={draw=none},
    ]
    \addplot+[fill=navy!70,draw=navy] coordinates {(0.99,0.896)};
    \addplot+[fill=amaranth!50,draw=amaranth] coordinates {(0.99,0.037)};
    \addplot+[fill=orange!70,draw=orange] coordinates {(0.99,0.398)};
    \end{axis}
    \end{tikzpicture}}
    
\subfloat[GIST]{
    \begin{tikzpicture}
    \begin{axis}[
        ybar stacked,
        height=\columnwidth/3,
        width=0.235\textwidth,
        bar width=9pt,
        symbolic x coords={0.95,0.99},
        xtick=data,
        xlabel=Target Recall@10,xlabel style={font=\scriptsize, yshift=4pt},
        ylabel=Latency (ms), ylabel style={yshift=-4pt},
        ymin=0, ymax=20.99,
        label style={font=\scriptsize},
        tick label style={font=\scriptsize},
        ymajorgrids=true, grid style=dashed,
        enlarge x limits=0.5,
        legend style={draw=none},
    ]
    \addplot+[fill=navy!70,draw=navy] coordinates {(0.95,1.656) (0.99,1.697)};
    \addplot+[fill=amaranth!50,draw=amaranth] coordinates {(0.95,0.297) (0.99,0.804)};
    \addplot+[fill=orange!70,draw=orange] coordinates {(0.95,7.122) (0.99,15.755)};
    \end{axis}
    \end{tikzpicture}}
\subfloat[MSMARCO-1M]{
    \begin{tikzpicture}
    \begin{axis}[
        ybar stacked,
        height=\columnwidth/3,
        width=0.235\textwidth,
        bar width=9pt,
        symbolic x coords={0.95,0.99},
        xtick=data,
        xlabel=Target Recall@10,xlabel style={font=\scriptsize, yshift=4pt},
        ylabel=Latency (ms),ylabel style={yshift=-4pt},
        ymin=0, ymax=6.58,
        label style={font=\scriptsize},
        tick label style={font=\scriptsize},
        ymajorgrids=true, grid style=dashed,
        enlarge x limits=0.5,
        legend style={draw=none},
    ]
    \addplot+[fill=navy!70,draw=navy] coordinates {(0.95,1.023) (0.99,1.095)};
    \addplot+[fill=amaranth!50,draw=amaranth] coordinates {(0.95,0.055) (0.99,0.374)};
    \addplot+[fill=orange!70,draw=orange] coordinates {(0.95,0.598) (0.99,4.253)};
    \end{axis}
    \end{tikzpicture}}
    \end{small}
    \vspace{-0.4cm}
    \caption{Query time breakdown at recall@10 targets (0.95 and 0.99).}
    \label{fig:phase-breakdown}\vspace{1em}
\end{figure}
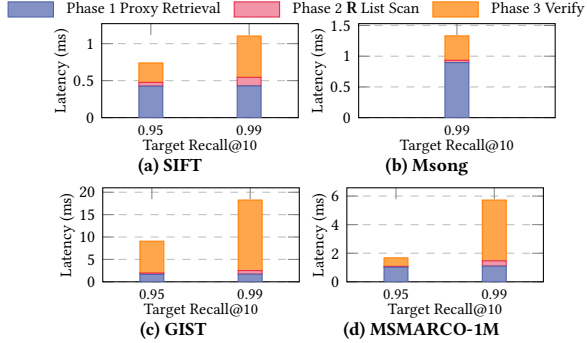

\stitle{Exp-3: Index Construction Time and Space Consumption.}
We next compare the index construction time and index size of different methods. The results are reported in Table~\ref{tab:index-time} and Table~\ref{tab:index-space}, respectively.
Table~\ref{tab:index-time} compares the construction cost of graph $\HNSW$, index $\HRNN$, and index $\HAMG$. 
Compared with the graph $\HNSW$, $\HRNN$ incurs additional offline overhead to construct the ranked $\KNN$ graph and materialize the reverse-neighbor lists. Nevertheless, it remains substantially cheaper to build than $\HAMG$. Across the four datasets, $\HAMG$ requires $5.2$--$14.4\times$ longer construction time than $\HRNN$, with the largest gap observed on MSMARCO.

Table~\ref{tab:index-space} reports the index size of different methods. The \emph{Base} column shows the raw vector storage size as a reference, while the method columns report index overhead only and exclude the raw vectors. $\HRNN$ stores both the navigation graph and the reverse-neighbor lists derived from the ranked $\KNN$ graph. As a result, it requires more memory than $\HNSW$. This overhead is most pronounced on low-dimensional datasets such as SIFT, where the raw vector storage is relatively small. On higher-dimensional datasets, the overhead becomes more moderate. For example, the total space consumption of $\HRNN$ is only $1.54\times$ and $1.52\times$ the raw vector size on GIST and MSMARCO, respectively.

 \begin{table}[t]
  \centering
  \small
    \caption{Index construction time (seconds).}\label{tab:index-time}

    \vspace{-0.4cm}
  
  \begin{tabular}{lrrr}
  \toprule
  Dataset & $\HNSW$ & $\HRNN$ & $\HAMG$ \\
  \midrule
  SIFT    & 49  & 185 & 954  \\
  Msong   & 91  & 280 & 1537 \\
  GIST    & 229 & 452 & 4286 \\
  MSMARCO & 199 & 553 & 7947 \\
  \bottomrule
  \end{tabular}
  \end{table}

  \begin{table}[t]
  \centering
  \small
  \setlength{\tabcolsep}{4pt}
  \caption{Index space consumption (MB).}
  \label{tab:index-space}

    \vspace{-0.4cm}
  
  \begin{tabular}{lrrrrr}
  \toprule
  Dataset & Base & $\HNSW$ & $\HAMG$ & $\HRNN$ & $\HRNN$ Total/Base \\
  \midrule
  SIFT    & 512.0  & 100.7 & 206.6 & 2112.6 & 5.13$\times$ \\
  Msong   & 1670.2 & 83.3  & 165.2 & 2083.4 & 2.25$\times$ \\
  GIST    & 3840.0 & 71.8  & 199.9 & 2083.8 & 1.54$\times$ \\
  MSMARCO & 4096.0 & 103.4 & 275.3 & 2114.9 & 1.52$\times$ \\
  \bottomrule
  \end{tabular}\vspace{1em}
  \end{table}

\stitle{Exp-4: Parameter Sensitivity.}
We study the impact of $\HRNN$'s two query-time parameters: the number of proxies $m$ and the reverse-neighbor threshold $\Theta$. Both parameters increase candidate coverage but affect different stages of the search process. Specifically, a larger $m$ expands the set of query-side proxies, while a larger $\Theta$ relaxes the rank constraint when scanning reverse-neighbor lists. Since the ranked $\KNN$ graph is constructed offline with a fixed size of $\mathbf{K}=500$, all settings satisfying $\Theta\leq\mathbf{K}$ are supported by the same index.
\ifthenelse{\isundefined{\themacro}}{
The full results are reported in the figure titled ``Recall@10 (top) and QPS (bottom, log scale) over the evaluated $(m,\Theta)$ parameter grid'' in our technical report~\cite{technicalreport}.
The results show that increasing either $m$ or $\Theta$ generally improves Recall@10 at the cost of lower throughput.}{}

Table~\ref{tab:param-selection} reports the highest-throughput configurations that satisfy Recall@10 thresholds of 0.95 and 0.99. The optimal parameter settings vary across datasets. For example, Msong saturates with $(m,\Theta)=(3,10)$, whereas GIST requires larger candidate coverage, reaching Recall@10 $\geq 0.99$ only with configurations such as $(300,50)$. SIFT and MSMARCO exhibit intermediate behavior.
These results reflect the degree of mismatch between nearby neighbors and R$k$NN results, with larger mismatches requiring larger values of $(m,\Theta)$. 
Importantly, $\HRNN$ can adapt to such variations by adjusting $(m,\Theta)$ at query time, without rebuilding the index.

\ifthenelse{\isundefined{\themacro}}{}{
Fig.~\ref{fig:appendix-exp-2} reports the complete $(m,\Theta)$ parameter grid evaluated in our experiments. Increasing either $m$ or $\Theta$ generally improves Recall@10, but reduces throughput. This trend is consistent across all datasets, reflecting the trade-off between candidate coverage and query efficiency.

\begin{figure*}
    \centering
    \includegraphics[width=\linewidth]{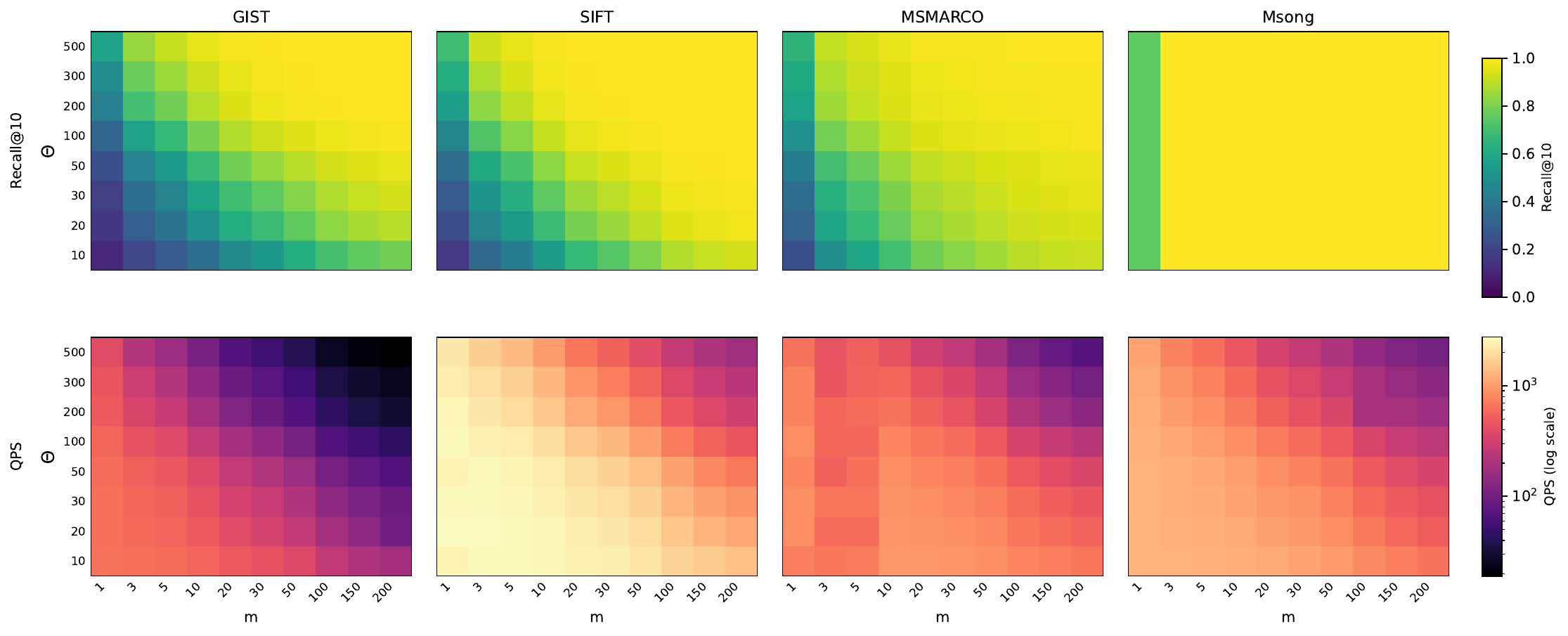}\vspace{-0.4cm}
    \caption{Recall@10 (top) and QPS (bottom, log scale) over the evaluated $(m,\Theta)$ parameter grid.}
    \label{fig:appendix-exp-2}
\end{figure*}}

\begin{table}[t]
\centering
\small
\caption{Best $\HRNN$ parameter configurations for target recall@10.}
\label{tab:param-selection}
\vspace{-0.4cm}

\begin{tabular}{lcccc}
\toprule
Dataset & \multicolumn{2}{c}{Recall@10 $\geq 0.95$} & \multicolumn{2}{c}{Recall@10 $\geq 0.99$} \\
 & $(m,\Theta)$ & Recall / QPS & $(m,\Theta)$ & Recall / QPS \\
\midrule
GIST     & $(10,500)$ & 0.9622 / 105.1  & $(50,300)$ & 0.9935 / 53.4 \\
SIFT     & $(20,100)$ & 0.9646 / 1561.8 & $(50,100)$ & 0.9916 / 1015.5 \\
MSMARCO  & $(30,100)$ & 0.9599 / 607.5  & $(50,500)$ & 0.9919 / 184.5 \\
Msong    & $(3,10)$   & 0.9990 / 1253.1 & $(3,10)$   & 0.9990 / 1253.1 \\
\bottomrule
\end{tabular}
\end{table}

\stitle{Exp-5: Ablation Study.}
We conduct ablation studies to evaluate the impact of the key components of $\HRNN$. Specifically, we examine three variants: (1) replacing $\HNSWG$ for proxy retrieval, (2) constructing the ranked $\KNNG$ without using the neighbor seeds collected from $\HNSWG$, and (3) performing candidate verification without pre-materialized $k$NN-radius information.

\sstitle{Navigation graph replacement.}
We first replace the $\HNSW$ graph used for proxy retrieval while keeping all other components unchanged.
Fig.~\ref{fig:navigation-graph-ablation} compares $\HNSW$ and $\kw{NSG}$. On GIST, the two graphs achieve similar performance in the high-recall region. At Recall@10 around 0.99, $\HNSW$ achieves 0.9906 recall at 63.9 QPS, while $\kw{NSG}$ achieves 0.9907 recall at 67.9 QPS. On SIFT, both graphs reach comparable recall, but $\HNSW$ delivers higher throughput. At Recall@10 around 0.99, $\HNSW$ achieves 0.9910 recall at 1441.8 QPS, compared with 0.9917 recall at 1082.5 QPS for $\kw{NSG}$. These results indicate that $\HRNN$ is not tied to a specific navigation graph, although the choice of graph affects query latency.

\sstitle{KNN graph construction with $\HNSW$-Seeding.}
We next evaluate how the ranked $\KNN$ graph is initialized prior to NNDescent refinement. Fig.~\ref{fig:knng-construction-time} compares standard NNDescent with random initialization against $\HRNN$'s $\HNSW$-seeded construction. On GIST, random initialization reaches only 0.9066 KNNG recall after 435\,s of refinement. In contrast, reusing $\HNSW$ insertion candidates as seeds yields an initial recall of 0.4880 and converges to perfect recall after 310.85\,s. A similar trend is observed on SIFT, where the final recall/construction time improves from 0.9874/387.47\,s to 1.0/272.51\,s. These results demonstrate that the $\HNSW$ insertion process provides high-quality initialization for NNDescent, reducing refinement cost while improving the quality of the resulting ranked $\KNN$ graph.

\sstitle{$k$NN-radius materialization.}
Finally, we evaluate whether $\HRNN$ can remove or approximate its materialized radius information. Table~\ref{tab:base-materialization-ablation} reports results on GIST with $m=20$, $\Theta=500$, and $k=10$.
The \emph{Gold Radius} variant replaces $\HRNN$'s $\KNNG$-derived radius estimates $\hat{r}_k(x)$ with exact radius $r_k(x)$ computed by brute force. Recall changes only marginally, from 0.9879 to 0.9880, indicating that the ranked $\KNN$ graph provides sufficiently accurate radius estimates for verification.
Removing the reverse-neighbor lists (i.e., the variant \emph{No Reverse-neighbor Lists}) forces $\HRNN$ to verify all 1M data points, increasing latency from 14.74\,ms to 372.04\,ms. This confirms that reverse-neighbor lists are essential for selective candidate generation and efficient verification.


\begin{figure}[t]
    \centering
    \begin{small}

    \begin{tikzpicture}
    \begin{customlegend}[legend columns=2,
    legend entries={HNSW (default), NSG},
    legend style={at={(0.5,1.20)},anchor=north,draw=none,font=\scriptsize,column sep=0.3cm}]
    \addlegendimage{line width=0.2mm, color=navy, mark=square, mark size=0.8mm}
    \addlegendimage{line width=0.2mm, color=amaranth, mark=triangle, mark size=0.8mm}
    \end{customlegend}
    \end{tikzpicture}
    \\[-\lineskip]

    \subfloat[GIST ($d\!=\!960$)]{\vspace{-2mm}
    \begin{tikzpicture}[scale=1]
    \begin{axis}[
        height=\columnwidth/3,
        width=0.235\textwidth,
        xlabel=Recall@10, xlabel style={yshift=+3pt},
        ylabel=QPS (1/s), ylabel style={yshift=-4pt},
        label style={font=\scriptsize},
        tick label style={font=\scriptsize},
        xmin=0.95, xmax=1.00,
        xtick={0.95,0.96,0.97,0.98,0.99,1.00},
        xticklabel style={/pgf/number format/fixed,/pgf/number format/precision=2,font=\scriptsize},
        ymode=log,
        ymajorgrids=true, xmajorgrids=true, grid style=dashed,
    ]
    \addplot[line width=0.2mm, color=navy, mark=square, mark size=0.6mm]
    plot coordinates {(0.9564, 112.73) (0.9583, 112.15) (0.9622, 104.66) (0.9658, 94.90) (0.9668, 78.09) (0.9859, 70.29) (0.9879, 67.75) (0.9906, 63.88) (0.9928, 56.52) (0.9951, 53.53) (0.9954, 51.39) (0.9965, 46.70) (0.9967, 40.40) (0.9992, 39.87) (0.9994, 39.15)};
    \addplot[line width=0.2mm, color=amaranth, mark=triangle, mark size=0.6mm]
    plot coordinates {(0.9678, 105.63) (0.9718, 103.45) (0.9725, 79.39) (0.9907, 67.92) (0.9909, 65.88) (0.9923, 62.71) (0.9938, 56.33) (0.9957, 51.70) (0.9959, 50.79) (0.9972, 45.71) (0.9982, 38.88) (0.9983, 37.83) (0.9994, 36.71) (0.9995, 34.52)};
    \end{axis}
    \end{tikzpicture}}
    \subfloat[SIFT ($d\!=\!128$)]{\vspace{-2mm}
    \begin{tikzpicture}[scale=1]
    \begin{axis}[
        height=\columnwidth/3,
        width=0.235\textwidth,
        xlabel=Recall@10, xlabel style={yshift=+3pt},
        ylabel=QPS (1/s), ylabel style={yshift=-4pt},
        label style={font=\scriptsize},
        tick label style={font=\scriptsize},
        xmin=0.96, xmax=1.00,
        xtick={0.96,0.97,0.98,0.99,1.00},
        xticklabel style={/pgf/number format/fixed,/pgf/number format/precision=2,font=\scriptsize},
        ymode=log,
        ymajorgrids=true, xmajorgrids=true, grid style=dashed,
    ]
    \addplot[line width=0.2mm, color=navy, mark=square, mark size=0.6mm]
    plot coordinates {(0.9615, 2547.1) (0.9736, 2124.0) (0.9811, 1753.8) (0.9855, 1717.6) (0.9859, 1527.2) (0.9910, 1441.8) (0.9924, 1332.1) (0.9948, 1130.3) (0.9963, 979.3) (0.9978, 927.5) (0.9990, 716.3) (0.9993, 596.4) (0.9996, 515.0)};
    \addplot[line width=0.2mm, color=amaranth, mark=triangle, mark size=0.6mm]
    plot coordinates {(0.9649, 1773.4) (0.9785, 1469.4) (0.9862, 1249.8) (0.9879, 1225.7) (0.9899, 1106.4) (0.9917, 1082.5) (0.9939, 967.0) (0.9965, 823.1) (0.9976, 748.6) (0.9983, 721.9) (0.9993, 572.7) (0.9996, 483.7) (0.9998, 419.0)};
    \end{axis}
    \end{tikzpicture}}
    \vspace{-0.4cm}
    \caption{Effect of navigation graph choice ($\mathbf{HNSW}$ vs.\ $\mathbf{NSG}$).}
    \label{fig:navigation-graph-ablation}\vspace{-0.5em}
    \end{small}
    
\end{figure}
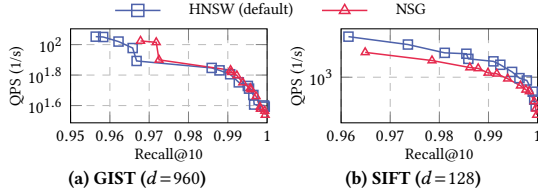


\begin{figure}[t]
    \centering
    \begin{small}

    \begin{tikzpicture}
    \begin{customlegend}[legend columns=2,
    legend entries={Random init + NNDescent, HNSW seed + NNDescent},
    legend style={at={(0.5,1.20)},anchor=north,draw=none,font=\scriptsize,column sep=0.3cm}]
    \addlegendimage{line width=0.2mm, color=orange, mark=pentagon, mark size=0.8mm}
    \addlegendimage{line width=0.2mm, color=amaranth, mark=triangle, mark size=0.8mm}
    \end{customlegend}
    \end{tikzpicture}
    \\[-0.8em]

    \subfloat[GIST ($d\!=\!960$)]{\vspace{-2mm}
    \begin{tikzpicture}[scale=1]
    \begin{axis}[
        height=\columnwidth/3,
        width=0.235\textwidth,
        xmin=0, xmax=450,
        ymin=0, ymax=1.05,
        xlabel=NNDescent time (s), xlabel style={yshift=+3pt},
        ylabel=KNNG recall@10, ylabel style={yshift=-4pt},
        label style={font=\scriptsize},
        tick label style={font=\scriptsize},
        ymajorgrids=true, xmajorgrids=true, grid style=dashed,
    ]
    \addplot[line width=0.2mm, color=orange, mark=pentagon, mark size=0.6mm]
    plot coordinates {(0.00, 0.0000) (28.69, 0.0016) (87.02, 0.0193) (184.57, 0.4493) (273.45, 0.7895) (353.81, 0.8713) (435.24, 0.9066)};
    \addplot[line width=0.2mm, color=amaranth, mark=triangle, mark size=0.6mm]
    plot coordinates {(0.00, 0.4880) (16.92, 0.6242) (49.15, 0.6958) (115.04, 0.8643) (196.46, 0.9706) (260.00, 0.9951) (310.85, 1.0000)};
    \end{axis}
    \end{tikzpicture}}
    \subfloat[SIFT ($d\!=\!128$)]{\vspace{-2mm}
    \begin{tikzpicture}[scale=1]
    \begin{axis}[
        height=\columnwidth/3,
        width=0.235\textwidth,
        xmin=0, xmax=400,
        ymin=0, ymax=1.05,
        xlabel=NNDescent time (s), xlabel style={yshift=+3pt},
        ylabel=KNNG recall@10, ylabel style={yshift=-4pt},
        label style={font=\scriptsize},
        tick label style={font=\scriptsize},
        ymajorgrids=true, xmajorgrids=true, grid style=dashed,
    ]
    \addplot[line width=0.2mm, color=orange, mark=pentagon, mark size=0.6mm]
    plot coordinates {(0.00, 0.0000) (29.43, 0.0015) (88.86, 0.0194) (180.48, 0.6128) (256.51, 0.9769) (320.10, 0.9863) (387.47, 0.9874)};
    \addplot[line width=0.2mm, color=amaranth, mark=triangle, mark size=0.6mm]
    plot coordinates {(0.00, 0.5072) (16.15, 0.7713) (52.61, 0.8684) (116.48, 0.9830) (180.33, 0.9999) (225.49, 1.0000) (272.51, 1.0000)};
    \end{axis}
    \end{tikzpicture}}
    \vspace{-0.4cm}
    \caption{Effect of HNSW-seeding on $\mathbf{KNN}$ graph construction.}
    \label{fig:knng-construction-time}
    \end{small}\vspace{-0.5em}
\end{figure}
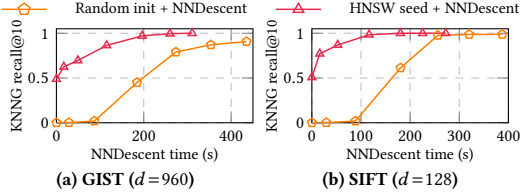

\begin{table}[t]
\centering
\small
\setlength{\tabcolsep}{4pt}
\caption{$k$NN-radius materialization ablation on GIST-1M.}
\label{tab:base-materialization-ablation}
\vspace{-0.4cm}
\begin{tabular}{lrrr}
\toprule
Variant & Recall@10 & Latency (ms) & Avg. Results \\
\midrule
$\HRNN$ & 0.9879 & 14.74 & 9.89 \\
Gold Radius & 0.9880 & 14.39 & 9.84 \\
No reverse-neighbor Lists & 0.9998 & 372.04 & 9.96 \\
\bottomrule
\end{tabular}\vspace{0.2em}
\end{table}

\stitle{Exp-6: Varying $k$.}
We evaluate $\HRNN$ under standard R$k$NN settings with varying values of $k \in {1,10,30,50,90}$.
For each dataset and each value of $k$, we evaluate the same $(m,\Theta)$ parameter grid used in Exp-1 and report the highest-throughput configuration satisfying target Recall@$k$ thresholds of 0.95 and 0.99. 

Fig.~\ref{fig:variable-k} shows that $\HRNN$ maintains stable recall--throughput trade-offs across the entire range of $k$. At the 0.99 recall target, increasing $k$ from 1 to 90 reduces throughput from 1245 to 400 QPS on SIFT, from 61 to 36 QPS on GIST, and from 279 to 105 QPS on MSMARCO. On Msong, the Recall@$k \geq 0.95$ and Recall@$k \geq 0.99$ curves coincide because the lowest-latency configuration already achieves Recall@$k \approx 0.999$ for all evaluated values of $k$, while maintaining a throughput of approximately 800--930 QPS.
These results demonstrate that $\HRNN$ remains robust across a wide range of target $k$ values. 
Cross-method comparisons are reported separately in Exp-1 under the common setting of $k=10$.


\begin{figure}[t]
    \centering
    \begin{small}
    \begin{tikzpicture}
    \begin{customlegend}[legend columns=4,
    legend entries={SIFT, GIST, Msong, MSMARCO},
    legend style={at={(0.5,1.18)},anchor=north,draw=none,font=\scriptsize,column sep=0.15cm}]
    \addlegendimage{line width=0.2mm, color=navy, mark=square, mark size=0.8mm}
    \addlegendimage{line width=0.2mm, color=amaranth, mark=triangle, mark size=0.8mm}
    \addlegendimage{line width=0.2mm, color=forestgreen, mark=o, mark size=0.8mm}
    \addlegendimage{line width=0.2mm, color=orange, mark=diamond, mark size=0.8mm}
    \end{customlegend}
    \end{tikzpicture}
    \\[-\lineskip]

    \subfloat[Target Recall@$k$: 0.95]{\vspace{-2mm}
    \begin{tikzpicture}[scale=1]
    \begin{axis}[
        height=\columnwidth/2.750,
        width=0.48\columnwidth,
        xmode=log,
        ymode=log,
        xmin=0.8, xmax=110,
        ymin=30, ymax=2200,
        xtick={1,10,30,50,90},
        xticklabels={1,10,30,50,90},
        log ticks with fixed point,
        xlabel={$k$}, xlabel style={yshift=+3pt},
        ylabel=QPS (1/s), ylabel style={yshift=-4pt},
        label style={font=\scriptsize},
        tick label style={font=\scriptsize},
        ymajorgrids=true, xmajorgrids=true, grid style=dashed,
    ]
    \addplot[line width=0.2mm, color=navy, mark=square, mark size=0.6mm]
    plot coordinates {(1,1672) (10,1281) (30,954) (50,1010) (90,859)};
    \addplot[line width=0.2mm, color=amaranth, mark=triangle, mark size=0.6mm]
    plot coordinates {(1,178) (10,95) (30,64) (50,61) (90,63)};
    \addplot[line width=0.2mm, color=forestgreen, mark=o, mark size=0.6mm]
    plot coordinates {(1,805) (10,934) (30,934) (50,901) (90,798)};
    \addplot[line width=0.2mm, color=orange, mark=diamond, mark size=0.6mm]
    plot coordinates {(1,658) (10,483) (30,531) (50,499) (90,344)};
    \end{axis}
    \end{tikzpicture}}
    \subfloat[Target Recall@$k$: 0.99]{\vspace{-2mm}
    \begin{tikzpicture}[scale=1]
    \begin{axis}[
        height=\columnwidth/2.750,
        width=0.48\columnwidth,
        xmode=log,
        ymode=log,
        xmin=0.8, xmax=110,
        ymin=30, ymax=2200,
        xtick={1,10,30,50,90},
        xticklabels={1,10,30,50,90},
        log ticks with fixed point,
        xlabel={$k$}, xlabel style={yshift=+3pt},
        ylabel=QPS (1/s), ylabel style={yshift=-4pt},
        label style={font=\scriptsize},
        tick label style={font=\scriptsize},
        ymajorgrids=true, xmajorgrids=true, grid style=dashed,
    ]
    \addplot[line width=0.2mm, color=navy, mark=square, mark size=0.6mm]
    plot coordinates {(1,1245) (10,866) (30,594) (50,603) (90,400)};
    \addplot[line width=0.2mm, color=amaranth, mark=triangle, mark size=0.6mm]
    plot coordinates {(1,61) (10,42) (30,41) (50,36) (90,36)};
    \addplot[line width=0.2mm, color=forestgreen, mark=o, mark size=0.6mm]
    plot coordinates {(1,805) (10,934) (30,934) (50,901) (90,798)};
    \addplot[line width=0.2mm, color=orange, mark=diamond, mark size=0.6mm]
    plot coordinates {(1,279) (10,174) (30,148) (50,105) (90,105)};
    \end{axis}
    \end{tikzpicture}}
    \vspace{-0.4cm}
    \caption{Performance of $\mathbf{HRNN}$ across different $k$ values.}
    \label{fig:variable-k}
    \end{small}
\end{figure}
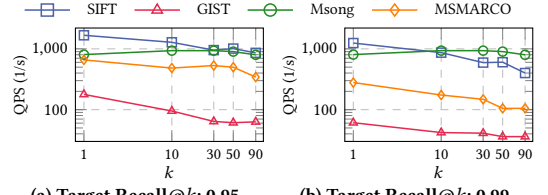

\stitle{Exp-7: Insertion-Based Maintenance.}
We evaluate insertion-based maintenance by varying the initial batch fraction $s$: the first $s\cdot n$ points are batch-built and the remaining points are inserted. 
Thus $s=1$ is pure batch construction and $s=0$ is pure insertion construction. For each $s$, we report the highest-QPS operating point reaching each target Recall@10.

Fig.~\ref{fig:insertion-pareto} shows that insertion-based maintenance largely preserves the recall-QPS trade-off. 
On SIFT, pure insertion reaches 1429, 1266, and 775 QPS at Recall@10 targets $0.90$, $0.95$, and $0.99$, respectively, compared with 1786, 1370, and 813 QPS for pure batch construction (that is, the index is constructed once over the entire dataset). 
On GIST, pure insertion is comparable to pure batch at the $0.90$ and $0.95$ targets, reaching 140/89 QPS versus 133/86 QPS. The maintained index thus does not collapse as the insertion fraction grows.

The main overhead is construction time, and this overhead is expected because insertion-based construction synchronizes multiple materialized views after each arriving point. Building entirely through insertions costs 155s versus 86s on SIFT and 760s versus 170s on GIST, i.e., $1.80\times$ and $4.47\times$ more than pure batch construction. This cost comes from maintaining the $\HNSW$ graph, ranked $\KNN$ graph, materialized radius estimates, and reverse-neighbor lists during insertion. Overall, $\HRNN$ supports continuous arrivals by paying additional write-side maintenance cost while keeping search efficiency stable.


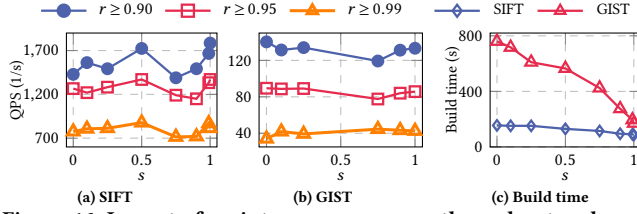
\begin{figure}[t]
    \centering
    \begin{small}

    \newlength{\figcolw}
    \setlength{\figcolw}{\columnwidth}

    \makebox[0.66\figcolw][c]{%
    \begin{tikzpicture}
    \begin{customlegend}[
        legend columns=3,
        legend entries={$r\!\geq\!0.90$,$r\!\geq\!0.95$,$r\!\geq\!0.99$},
        legend style={
            at={(0,0)},
            anchor=center,
            draw=none,
            font=\scriptsize,
            column sep=0.08cm
        }]
    \addlegendimage{line width=0.3mm, color=navy, mark=*, mark size=0.8mm}
    \addlegendimage{line width=0.3mm, color=amaranth, mark=square, mark size=0.8mm}
    \addlegendimage{line width=0.4mm, color=orange, mark=triangle, mark size=0.8mm}
    \end{customlegend}
    \end{tikzpicture}}%
    \makebox[0.31\figcolw][c]{%
    \begin{tikzpicture}
    \begin{customlegend}[
        legend columns=2,
        legend entries={SIFT,GIST},
        legend style={
            at={(0,0)},
            anchor=center,
            draw=none,
            font=\scriptsize,
            column sep=0.1cm
        }]
    \addlegendimage{line width=0.3mm, color=navy, mark=diamond, mark size=0.8mm}
    \addlegendimage{line width=0.3mm, color=amaranth, mark=triangle, mark size=0.8mm}
    \end{customlegend}
    \end{tikzpicture}}%

    \vspace{-1.0mm}

    \noindent
    \begin{minipage}[t]{0.32\figcolw}
    \centering
    \begin{tikzpicture}
    \begin{axis}[
        scale only axis,
        width=0.235\figcolw,
        height=\columnwidth/1.8,
        xlabel={$s$},
        xlabel style={yshift=+3pt},
        ylabel={QPS (1/s)},
        ylabel style={yshift=-6pt},
        xmin=-0.05, xmax=1.05,
        ymin=600, ymax=1900,
        ytick={700,1200,1700},
        label style={font=\scriptsize},
        tick label style={font=\scriptsize},
        ymajorgrids=true,
        xmajorgrids=true,
        grid style=dashed,
    ]
    \addplot[color=navy,line width=0.3mm,solid,mark=*,mark size=0.7mm]
    coordinates {
        (0.00,1428.6) (0.10,1562.5) (0.25,1492.5) (0.50,1724.1)
        (0.75,1388.9) (0.90,1492.5) (0.99,1666.7) (1.00,1785.7)};
    \addplot[color=amaranth,line width=0.3mm,mark=square,mark size=0.7mm]
    coordinates {
        (0.00,1265.8) (0.10,1219.5) (0.25,1282.1) (0.50,1369.9)
        (0.75,1190.5) (0.90,1149.4) (0.99,1333.3) (1.00,1369.9)};
    \addplot[color=orange,line width=0.4mm,solid,mark=triangle,mark size=0.9mm]
    coordinates {
        (0.00,775.2) (0.10,806.5) (0.25,813.0) (0.50,877.2)
        (0.75,714.3) (0.90,719.4) (0.99,869.6) (1.00,813.0)};
    \end{axis}
    \end{tikzpicture}

    \vspace{-1.5mm}
    {\scriptsize\textbf{(a) SIFT}}
    \end{minipage}%
    \hfill
    \begin{minipage}[t]{0.32\figcolw}
    \centering
    \begin{tikzpicture}
    \begin{axis}[
        scale only axis,
        width=0.255\figcolw,
        height=\columnwidth/1.8,
        xlabel={$s$},
        xlabel style={yshift=+3pt},
        xmin=-0.05, xmax=1.05,
        ymin=25, ymax=150,
        ytick={40,80,120},
        label style={font=\scriptsize},
        tick label style={font=\scriptsize},
        ymajorgrids=true,
        xmajorgrids=true,
        grid style=dashed,
    ]
    \addplot[color=navy,line width=0.3mm,solid,mark=*,mark size=0.7mm]
    coordinates {
        (0.00,140.3) (0.10,131.4) (0.25,134.0)
        (0.75,119.3) (0.90,131.2) (1.00,133.3)};
    \addplot[color=amaranth,line width=0.3mm,mark=square,mark size=0.7mm]
    coordinates {
        (0.00,89.4) (0.10,88.7) (0.25,89.1)
        (0.75,77.6) (0.90,84.0) (1.00,85.7)};
    \addplot[color=orange,line width=0.4mm,solid,mark=triangle,mark size=0.9mm]
    coordinates {
        (0.00,34.0) (0.10,41.8) (0.25,39.3)
        (0.75,44.5) (0.90,43.3) (1.00,42.5)};
    \end{axis}
    \end{tikzpicture}

    \vspace{-1.5mm}
    {\scriptsize\textbf{(b) GIST}}
    \end{minipage}%
    \hfill
    \begin{minipage}[t]{0.32\figcolw}
    \centering
    \begin{tikzpicture}
    \begin{axis}[
        scale only axis,
        width=0.235\figcolw,
        height=\columnwidth/1.8,
        xlabel={$s$},
        xlabel style={yshift=+3pt},
        ylabel={Build time (s)},
        ylabel style={yshift=-6pt},
        xmin=-0.05, xmax=1.05,
        ymin=0, ymax=820,
        ytick={0,400,800},
        label style={font=\scriptsize},
        tick label style={font=\scriptsize},
        ymajorgrids=true,
        xmajorgrids=true,
        grid style=dashed,
    ]
    \addplot[color=navy,line width=0.3mm,solid,mark=diamond,mark size=0.7mm]
    coordinates {
        (0.00,155) (0.10,152) (0.25,152) (0.50,130)
        (0.75,115) (0.90,95) (0.99,91) (1.00,86)};
    \addplot[color=amaranth,line width=0.3mm,mark=triangle,mark size=0.9mm]
    coordinates {
        (0.00,760) (0.10,717) (0.25,610) (0.50,565)
        (0.75,424) (0.90,275) (0.99,195) (1.00,170)};
    \end{axis}
    \end{tikzpicture}

    \vspace{-1.5mm}
    {\scriptsize\textbf{(c) Build time}}
    \end{minipage}

    \vspace{-0.4cm}
    \caption{Impact of maintenance on query throughput and construction time. Panels (a)--(b) show throughput at Recall@10 targets of 0.90, 0.95, and 0.99; panel (c) shows total construction time.}
    \label{fig:insertion-pareto}
    
    \end{small}
\end{figure}


\begin{figure}[t]
    
    \centering
    \begin{small}
    \begin{tikzpicture}
    \begin{customlegend}[legend columns=4, legend entries={HRNN, HAMG, HNSW-SFT, HNSW-RDT}, legend style={at={(0.5,1.15)},anchor=north,draw=none,font=\scriptsize,column sep=0.1cm}]
    \addlegendimage{line width=0.2mm, color=navy, mark=square, mark size=0.8mm}
    \addlegendimage{line width=0.2mm, color=amaranth, mark=triangle, mark size=0.8mm}
    \addlegendimage{line width=0.2mm, color=blue, mark=diamond, mark size=0.8mm}
    \addlegendimage{line width=0.2mm, color=orange, mark=pentagon, mark size=0.8mm}
    \end{customlegend}
    \end{tikzpicture}
    \\[-0.85em]
    \subfloat[SIFT (d=128)]{\vspace{-2mm}
    \begin{tikzpicture}[scale=1]
    \begin{axis}[
        height=\columnwidth/2.70,
        width=0.42\columnwidth, %
        ymode=log,
        xmin=1.5, xmax=10.5,
        xtick={2,4,6,8,10},
        xticklabels={2M,4M,6M,8M,10M},
        ymin=1.076, ymax=2399,
        xlabel={$|D|$ (M)}, xlabel style={yshift=+3pt},
        ylabel={QPS (1/s)}, ylabel style={yshift=-4pt},
        label style={font=\scriptsize},
        tick label style={font=\scriptsize},
        ymajorgrids=true, xmajorgrids=true, grid style=dashed,
    ]
    \addplot[line width=0.2mm, color=navy, mark=square, mark size=0.65mm]
    plot coordinates {(2, 1370.8) (4, 901.7) (6, 834.3) (8, 995) (10, 842)};
    \addplot[line width=0.2mm, color=amaranth, mark=triangle, mark size=0.65mm]
    plot coordinates {(2, 6.7828) (4, 5.7275) (6, 1.9559)};
    \addplot[line width=0.2mm, color=blue, mark=diamond, mark size=0.65mm]
    plot coordinates {(2, 862.366) (4, 512.899) (6, 480.307) (8, 365.123) (10, 225.164)};
    \addplot[line width=0.2mm, color=orange, mark=pentagon, mark size=0.65mm]
    plot coordinates {(2, 131.742) (4, 107.562) (6, 94.141) (8, 100.282) (10, 70.992)};
    \end{axis}
    \end{tikzpicture}}
    \subfloat[MSMARCO (d=1024)]{\vspace{-2mm}
    \begin{tikzpicture}[scale=1]
    \begin{axis}[
        height=\columnwidth/2.70,
        width=0.42\columnwidth, %
        ymode=log,
        xmin=1.5, xmax=10.5,
        xtick={2,4,6,8,10},
        xticklabels={2M,4M,6M,8M,10M},
        ymin=0.7836, ymax=997.9,
        xlabel={$|D|$ (M)}, xlabel style={yshift=+3pt},
        ylabel={QPS (1/s)}, ylabel style={yshift=-4pt},
        label style={font=\scriptsize},
        tick label style={font=\scriptsize},
        ymajorgrids=true, xmajorgrids=true, grid style=dashed,
    ]
    \addplot[line width=0.2mm, color=navy, mark=square, mark size=0.65mm]
    plot coordinates {(2, 570.2) (4, 510.2) (6, 488.9) (8, 406.3) (10, 360.2)};
    \addplot[line width=0.2mm, color=amaranth, mark=triangle, mark size=0.65mm]
    plot coordinates {(2, 1.6686) (4, 1.4247)};
    \addplot[line width=0.2mm, color=blue, mark=diamond, mark size=0.65mm]
    plot coordinates {(2, 28.913) (4, 35.864) (6, 35.277) (8, 32.624) (10, 32.287)};
    \addplot[line width=0.2mm, color=orange, mark=pentagon, mark size=0.65mm]
    plot coordinates {(2, 11.286) (4, 13.259) (6, 13.752) (8, 11.692) (10, 11.294)};
    \end{axis}
    \end{tikzpicture}}
    \vspace{-0.4cm}
    \caption{Cross-method scalability at recall@10 $\geq 0.95$.}
    \label{fig:scalability-method-qps}
    \end{small}
    \vspace{0.5em}
\end{figure}
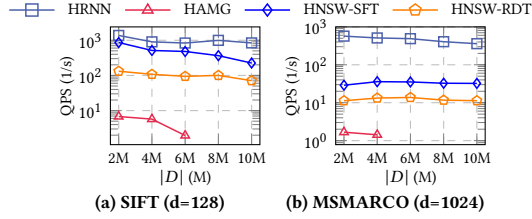

\stitle{Exp-8: Scalability to 10M Vectors.}
We evaluate scalability on the SIFT and MSMARCO datasets with sizes $|D| \in \{2,4,6,8,10\}$M. We compare $\HRNN$, $\HNSW$-$\SFT$, and $\HNSW$-$\RDT$ across all dataset sizes. For $\HAMG$, we report only completed runs; larger instances are omitted when construction time exceeds $10^4$ seconds and query throughput remains below 10 QPS.
Fig.\ref{fig:scalability-method-qps} compares query throughput at a common target of Recall@10 $\geq 0.95$. On SIFT, $\HRNN$ is $1.59$--$3.74\times$ faster than $\HNSW$-$\SFT$ and $8.4$--$11.9\times$ faster than $\HNSW$-$\RDT$ across all dataset sizes. The advantage over $\HAMG$ is even more pronounced. For example, on SIFT-6M, $\HRNN$ achieves 834.3 QPS, whereas $\HAMG$ achieves only 1.96 QPS.
These results demonstrate the superior scalability of $\HRNN$ on large-scale datasets. 
\ifthenelse{\isundefined{\themacro}}{Additional scalability results are provided in our technical report\cite{technicalreport}.}{}

\ifthenelse{\isundefined{\themacro}}{}{
Fig.~\ref{fig:scalability-pareto} shows that $\HRNN$ maintains a stable recall--throughput trade-off up to 10M vectors. At Recall@10 $\geq 0.95$, $\HRNN$ achieves 842 QPS on SIFT-10M and 360.2 QPS on MSMARCO-10M. It further reaches Recall@10 = 0.9975 on SIFT-10M while sustaining 211.7 QPS.

Fig.~\ref{fig:scalability-method-build} reports index construction time as the dataset size increases. On SIFT-10M, $\HRNN$ builds the index in 2335\,s, compared with 559\,s for $\HNSW$. On MSMARCO-10M, $\HRNN$ requires 5006\,s, compared with 2433\,s for $\HNSW$. This additional offline cost stems from constructing the ranked $\KNN$ graph and materializing the reverse-neighbor lists. Nevertheless, it remains substantially lower than that of $\HAMG$, which already requires 11405\,s on SIFT-6M and 37837\,s on MSMARCO-4M.


\providecolor{sizeViridis2M}{RGB}{68,  1, 84}
\providecolor{sizeViridis4M}{RGB}{59, 82,139}
\providecolor{sizeViridis6M}{RGB}{33,144,140}
\providecolor{sizeViridis8M}{RGB}{93,201, 99}
\providecolor{sizeViridis10M}{RGB}{255, 166, 76}

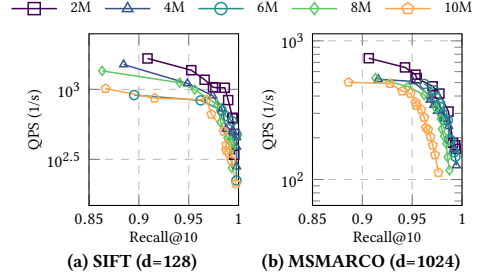
\begin{figure}[t]
    \centering
    \begin{small}
    \begin{tikzpicture}
    \begin{customlegend}[legend columns=5, legend entries={2M, 4M, 6M, 8M, 10M}, legend style={at={(0.5,1.15)},anchor=north,draw=none,font=\scriptsize,column sep=0.1cm}]
    \addlegendimage{line width=0.2mm, color=sizeViridis2M, mark=square, mark size=0.8mm}
    \addlegendimage{line width=0.2mm, color=sizeViridis4M, mark=triangle, mark size=0.8mm}
    \addlegendimage{line width=0.2mm, color=sizeViridis6M, mark=o, mark size=0.8mm}
    \addlegendimage{line width=0.2mm, color=sizeViridis8M, mark=diamond, mark size=0.8mm}
    \addlegendimage{line width=0.2mm, color=sizeViridis10M, mark=pentagon, mark size=0.8mm}
    \end{customlegend}
    \end{tikzpicture}
    \\[-\lineskip]
    \subfloat[SIFT (d=128)]{\vspace{-2mm}
    \begin{tikzpicture}[scale=1]
    \begin{axis}[
        height=\columnwidth/2.20,
        width=0.42\columnwidth, %
        xmin=0.85, xmax=1.0,
        ymin=148.19, ymax=2493.45,
        ymode=log,
        xlabel=Recall@10, xlabel style={yshift=+3pt},
        ylabel=QPS (1/s), ylabel style={yshift=-4pt},
        label style={font=\scriptsize},
        tick label style={font=\scriptsize},
        ymajorgrids=true, xmajorgrids=true, grid style=dashed,
    ]
    \addplot[line width=0.2mm, color=sizeViridis2M, mark=square, mark size=0.6mm]
    plot coordinates {(0.9087, 1662.300) (0.9526, 1370.800) (0.9660, 1170.900) (0.9744, 1036.200) (0.9768, 1029.100) (0.9862, 1026.000) (0.9896, 838.300) (0.9940, 626.900) (0.9940, 620.400) (0.9957, 481.000) (0.9957, 340.900)};
    \addplot[line width=0.2mm, color=sizeViridis4M, mark=triangle, mark size=0.6mm]
    plot coordinates {(0.8847, 1504.900) (0.9488, 1101.900) (0.9741, 901.700) (0.9833, 728.400) (0.9877, 521.100) (0.9936, 502.400) (0.9939, 457.400) (0.9980, 454.800) (0.9980, 389.300) (0.9980, 280.900)};
    \addplot[line width=0.2mm, color=sizeViridis6M, mark=o, mark size=0.6mm]
    plot coordinates {(0.8953, 910.100) (0.9620, 834.300) (0.9831, 706.400) (0.9873, 645.100) (0.9934, 593.800) (0.9979, 476.400) (0.9979, 223.500)};
    \addplot[line width=0.2mm, color=sizeViridis8M, mark=diamond, mark size=0.6mm]
    plot coordinates {(0.8632, 1354.700) (0.9409, 1115.800) (0.9561, 995.000) (0.9786, 778.500) (0.9820, 583.900) (0.9887, 450.500) (0.9910, 429.400) (0.9932, 413.500) (0.9932, 379.000) (0.9932, 274.500)};
    \addplot[line width=0.2mm, color=sizeViridis10M, mark=pentagon, mark size=0.6mm]
    plot coordinates {(0.8666, 1012.900) (0.9158, 864.900) (0.9679, 842.000) (0.9728, 664.000) (0.9852, 500.900) (0.9877, 400.700) (0.9877, 366.400) (0.9901, 351.800) (0.9926, 312.200) (0.9975, 211.700)};
    \end{axis}
    \end{tikzpicture}}
    \subfloat[MSMARCO (d=1024)]{\vspace{-2mm}
    \begin{tikzpicture}[scale=1]
    \begin{axis}[
        height=\columnwidth/2.20,
        width=0.42\columnwidth, %
        xmin=0.85, xmax=1.0,
        ymin=64.96, ymax=1124.7,
        ymode=log,
        xlabel=Recall@10, xlabel style={yshift=+3pt},
        ylabel=QPS (1/s), ylabel style={yshift=-4pt},
        label style={font=\scriptsize},
        tick label style={font=\scriptsize},
        ymajorgrids=true, xmajorgrids=true, grid style=dashed,
    ]
    \addplot[line width=0.2mm, color=sizeViridis2M, mark=square, mark size=0.6mm]
    plot coordinates {(0.9065, 749.800) (0.9428, 643.700) (0.9533, 570.200) (0.9545, 537.500) (0.9705, 470.300) (0.9773, 416.700) (0.9864, 309.000) (0.9910, 183.600) (0.9932, 176.500) (0.9939, 166.400)};
    \addplot[line width=0.2mm, color=sizeViridis4M, mark=triangle, mark size=0.6mm]
    plot coordinates {(0.9160, 528.600) (0.9505, 510.200) (0.9677, 403.400) (0.9680, 375.800) (0.9718, 345.400) (0.9744, 344.500) (0.9754, 311.400) (0.9810, 289.300) (0.9825, 273.400) (0.9865, 242.400) (0.9922, 164.900) (0.9943, 126.200)};
    \addplot[line width=0.2mm, color=sizeViridis6M, mark=o, mark size=0.6mm]
    plot coordinates {(0.9623, 488.900) (0.9721, 438.600) (0.9735, 394.000) (0.9779, 352.400) (0.9781, 333.800) (0.9802, 298.600) (0.9821, 268.300) (0.9857, 267.200) (0.9871, 195.000) (0.9888, 192.300) (0.9919, 157.100) (0.9934, 146.100)};
    \addplot[line width=0.2mm, color=sizeViridis8M, mark=diamond, mark size=0.6mm]
    plot coordinates {(0.9131, 539.100) (0.9456, 468.600) (0.9621, 406.300) (0.9737, 326.600) (0.9786, 273.500) (0.9797, 230.700) (0.9817, 205.700) (0.9841, 181.200) (0.9854, 145.000) (0.9873, 116.600)};
    \addplot[line width=0.2mm, color=sizeViridis10M, mark=pentagon, mark size=0.6mm]
    plot coordinates {(0.8862, 502.300) (0.9277, 492.200) (0.9400, 434.200) (0.9484, 416.900) (0.9534, 360.200) (0.9546, 341.600) (0.9592, 319.200) (0.9621, 280.300) (0.9638, 237.100) (0.9660, 228.500) (0.9684, 203.000) (0.9707, 176.600) (0.9729, 150.500) (0.9763, 112.200)};
    \end{axis}
    \end{tikzpicture}}
    \vspace{-0.4cm}
    \caption{Scalability of $\HRNN$ across dataset sizes.}
    \label{fig:scalability-pareto}
    \end{small}
\end{figure}

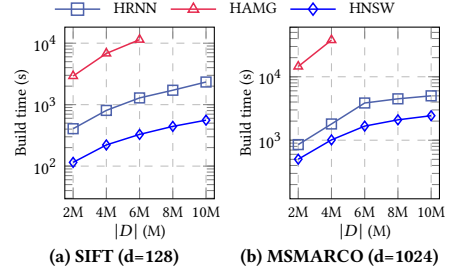
\begin{figure}[t]
    \centering
    \begin{small}
    \begin{tikzpicture}
    \begin{customlegend}[legend columns=3, legend entries={HRNN, HAMG, HNSW}, legend style={at={(0.5,1.15)},anchor=north,draw=none,font=\scriptsize,column sep=0.1cm}]
    \addlegendimage{line width=0.2mm, color=navy, mark=square, mark size=0.8mm}
    \addlegendimage{line width=0.2mm, color=amaranth, mark=triangle, mark size=0.8mm}
    \addlegendimage{line width=0.2mm, color=blue, mark=diamond, mark size=0.8mm}
    \end{customlegend}
    \end{tikzpicture}
    \\[-\lineskip]
    \subfloat[SIFT (d=128)]{\vspace{-2mm}
    \begin{tikzpicture}[scale=1]
    \begin{axis}[
        height=\columnwidth/2.20,
        width=0.42\columnwidth, %
        ymode=log,
        xmin=1.5, xmax=10.5,
        xtick={2,4,6,8,10},
        xticklabels={2M,4M,6M,8M,10M},
        ymin=29.4, ymax=1.825e+04,
        xlabel={$|D|$ (M)}, xlabel style={yshift=+3pt},
        ylabel={Build time (s)}, ylabel style={yshift=-4pt},
        label style={font=\scriptsize},
        tick label style={font=\scriptsize},
        ymajorgrids=true, xmajorgrids=true, grid style=dashed,
    ]
    \addplot[line width=0.2mm, color=navy, mark=square, mark size=0.65mm]
    plot coordinates {(2, 406) (4, 806) (6, 1287) (8, 1716) (10, 2335)};
    \addplot[line width=0.2mm, color=amaranth, mark=triangle, mark size=0.65mm]
    plot coordinates {(2, 2920) (4, 6813) (6, 11405)};
    \addplot[line width=0.2mm, color=blue, mark=diamond, mark size=0.65mm]
    plot coordinates {(2, 115) (4, 221) (6, 328) (8, 443) (10, 559)};
    \end{axis}
    \end{tikzpicture}}
    \subfloat[MSMARCO (d=1024)]{\vspace{-2mm}
    \begin{tikzpicture}[scale=1]
    \begin{axis}[
        height=\columnwidth/2.20,
        width=0.42\columnwidth, %
        ymode=log,
        xmin=1.5, xmax=10.5,
        xtick={2,4,6,8,10},
        xticklabels={2M,4M,6M,8M,10M},
        ymin=119.4, ymax=6.054e+04,
        xlabel={$|D|$ (M)}, xlabel style={yshift=+3pt},
        ylabel={Build time (s)}, ylabel style={yshift=-4pt},
        label style={font=\scriptsize},
        tick label style={font=\scriptsize},
        ymajorgrids=true, xmajorgrids=true, grid style=dashed,
    ]
    \addplot[line width=0.2mm, color=navy, mark=square, mark size=0.65mm]
    plot coordinates {(2, 853) (4, 1803) (6, 3860) (8, 4491) (10, 5006)};
    \addplot[line width=0.2mm, color=amaranth, mark=triangle, mark size=0.65mm]
    plot coordinates {(2, 14508) (4, 37837)};
    \addplot[line width=0.2mm, color=blue, mark=diamond, mark size=0.65mm]
    plot coordinates {(2, 502) (4, 1012) (6, 1669) (8, 2095) (10, 2433)};
    \end{axis}
    \end{tikzpicture}}
    \vspace{-0.4cm}
    \caption{Index construction scalability across dataset sizes.}
    \label{fig:scalability-method-build}
    \end{small}
\end{figure}
}

\section{Related Work}
As introduced in Section~\ref{sec:problem-analysis}, existing methods for (approximate) R$k$NN search generally follow a filter-and-verification framework. For most existing methods, the verification stage is similar: they issue a standard $k$NN search for each candidate to determine whether the candidate belongs to the R$k$NN result set of the query vector. Therefore, existing methods mainly differ in how they generate the candidate set. Based on this distinction, we categorize existing methods into three groups below.

\stitle{Classical Methods.}
Early studies mainly focused on exact R$k$NN search in low-dimensional vector spaces. One category is partition-based methods~\cite{TPL-2004-VLDB, TPL2-2007-Springer, SLICE-2014-ICDE, FINCH-2008-VLDB}, which divide the data space into geometric regions such as half-spaces or influence zones. These methods prune impossible candidates based on the spatial partition and the relative position between the query vector and the partitioned regions.
Another category is precomputation-based methods~\cite{MRKNNCop-2006-SIGMOD, ERkNN-2005-CIKM}. These methods typically use tree-based spatial indexes (e.g., M-trees) and derive conservative distance bounds for subtrees to prune unpromising points during query processing.
However, these classical methods suffer severely from the curse of dimensionality. In high-dimensional spaces, partition-based methods lose their discriminative power and often degenerate into scanning large portions of the dataset. Precomputation-based methods also incur high index construction and maintenance costs due to expensive distance-bound computation.

\stitle{Expansion-Based Methods.}
To scale R$k$NN queries to high-dimensional vectors, expansion-based methods such as $\SFT$~\cite{SFT-2003} and $\RDT$~\cite{RDT-2017-VLDB} were proposed. These methods aim to find approximate R$k$NN results by progressively exploring nearby neighbors around the query vector $q$ until a stopping condition is satisfied. As analyzed in Section~\ref{sec:problem-analysis}, these methods become inefficient when nearby neighbors do not align well with the true R$k$NN results, since they must explore a large neighborhood region to maintain high search accuracy.

\stitle{The State-of-the-Art and Beyond.}
The state-of-the-art method $\HAMG$~\cite{HAMG-ICDE-2024} adapts proximity graphs for R$k$NN search. $\HAMG$ leverages the key property of the $\kw{MRN}$ graph that guarantees true R$k$NN results lie within the $k$-hop neighborhood of the query vector on the graph. Based on this property, $\HAMG$ retrieves nearby graph neighbors as candidates. However, $\HAMG$ may still explore a large number of unnecessary candidates because the graph neighborhood can be substantially larger than the true R$k$NN result set.

Proximity graphs such as $\HNSW$~\cite{HNSW-2020-TPAMI} and $\kw{NSG}$~\cite{NSG-VLDB-2019-Fu} have become the dominant paradigm for high-dimensional approximate $k$NN search. 
$\HAMG$ elegantly exploits the properties of the $\kw{MRN}$ graph, which is also a proximity graph. 
Since exact construction of an $\kw{MRN}$ graph is expensive, $\kw{HAMG}$ heuristically adapts the $\HNSW$ graph to better approximate the $\kw{MRN}$ structure. Although this heuristic adaptation weakens the theoretical guarantees of the original $\kw{MRN}$ graph, it remains effective in practice.
There is a large body of work has studied high-dimensional $k$NN search from different perspectives, including distance computation optimization~\cite{ADSampling-2023-SIGMOD, MRQ-2026-VLDB, DDC-2025-ICDE, DADE-2024-VLDB}, novel index structures~\cite{SYMQG-2025-SIGMOD,Starling-2024-SIGMOD, AiSAQ-2025-arxiv}, and search optimization techniques~\cite{PipeANN-2026-OSDI, OptimizationsGraphbasedDiskresident--2026-li, FINGER-2023-WWW, ADA-NNS-2025-WWW, PEO-2024-ICML}. Since this paper focuses on R$k$NN search, we omit detailed a discussion of these methods.

\section{Conclusion}
This paper studies approximate R$k$NN search over high-dimensional vectors. Motivated by the limitations of existing methods, we propose $\HRNN$, a hybrid graph index that treats nearby neighbors as proxy points and materializes $k$NN-radius offline to avoid costly online computation. 
$\HRNN$ integrates a navigation graph, a ranked $\KNN$ graph, and reverse-neighbor lists, together with efficient index construction and append-only maintenance algorithms. 
Experiments on real-world datasets demonstrate that $\HRNN$ consistently outperforms existing methods in both efficiency and scalability while maintaining high search accuracy.

\ifthenelse{\isundefined{\themacro}}{\newpage}{}

\balance
\bibliographystyle{ACM-Reference-Format}
\bibliography{reference}

@String{Computing = "Computing" }

@String{Computer = "{IEEE} Computer" }

@ArtifactSoftware{R,
    title = {R: A Language and Environment for Statistical Computing},
    author = {{R Core Team}},
    organization = {R Foundation for Statistical Computing},
    address = {Vienna, Austria},
    year = {2019},
    url = {https://www.R-project.org/},
}

@article{LLM-RAG-NIPS-2020,
  title={Retrieval-augmented generation for knowledge-intensive nlp tasks},
  author={Lewis, Patrick and Perez, Ethan and Piktus, Aleksandra and Petroni, Fabio and Karpukhin, Vladimir and Goyal, Naman and K{\"u}ttler, Heinrich and Lewis, Mike and Yih, Wen-tau and Rockt{\"a}schel, Tim and others},
  journal={Advances in Neural Information Processing Systems},
  volume={33},
  pages={9459--9474},
  year={2020}
}

@article{MQA-VLDB-2024-Wang,
author = {Wang, Mengzhao and Wu, Haotian and Ke, Xiangyu and Gao, Yunjun and Xu, Xiaoliang and Chen, Lu},
title = {An Interactive Multi-Modal Query Answering System with Retrieval-Augmented Large Language Models},
year = {2024},
issue_date = {August 2024},
publisher = {VLDB Endowment},
volume = {17},
number = {12},
issn = {2150-8097},
url = {https://doi.org/10.14778/3685800.3685868},
doi = {10.14778/3685800.3685868},
journal = {Proc. VLDB Endow.},
month = aug,
pages = {4333–4336},
numpages = {4}
}

@INPROCEEDINGS{HAMG-ICDE-2024,
  author={Song, Yitong and Wang, Kai and Yao, Bin and Chen, Zhida and Xie, Jiong and Li, Feifei},
  booktitle={2024 IEEE 40th International Conference on Data Engineering (ICDE)}, 
  title={Efficient Reverse $k$ Approximate Nearest Neighbor Search Over High-Dimensional Vectors}, 
  year={2024},
  volume={},
  number={},
  pages={4262-4274},
  doi={10.1109/ICDE60146.2024.00325}}

@inproceedings{Rdnn-tree-2001-ICDE,
  title={An index structure for efficient reverse nearest neighbor queries},
  author={Yang, Congyun and Lin, King-Ip},
  booktitle={Proceedings 17th International Conference on Data Engineering},
  pages={485--492},
  year={2001},
  organization={IEEE}
}

@inproceedings{RKNN-Road-Network-SSTD-2025,
author = {Song, Anbang and Yu, Ziqiang and Liu, Wei and Xu, Yating and Tao, Mingjin},
title = {BRkNN-light: Batch Processing of Reverse k-Nearest Neighbor Queries for Moving Objects on Road Networks},
year = {2025},
publisher = {Association for Computing Machinery},
address = {New York, NY, USA},
url = {https://doi.org/10.1145/3748777.3748791},
doi = {10.1145/3748777.3748791},
booktitle = {Proceedings of the 19th International Symposium on Spatial and Temporal Data},
pages = {80–89},
numpages = {10},
location = {
},
series = {SSTD '25}
}

@article{RNN-Outlier-Detection-2025-PatAna,
author = {Heydari-Gharaei, Reza and Sharifi, Rasoul and Kashef, Shima and Nezamabadi-pour, Hossein},
title = {A robust approach for outlier detection based on the ratio of number of reverse neighbors to neighbors},
year = {2025},
issue_date = {Mar 2025},
publisher = {Springer-Verlag},
address = {Berlin, Heidelberg},
volume = {28},
number = {1},
issn = {1433-7541},
url = {https://doi.org/10.1007/s10044-024-01372-y},
doi = {10.1007/s10044-024-01372-y},
journal = {Pattern Anal. Appl.},
month = jan,
numpages = {30}
}

@article{OutlierDetection-2015-PatRecLet,
author = {Bhattacharya, Gautam and Ghosh, Koushik and Chowdhury, Ananda S.},
title = {Outlier detection using neighborhood rank difference},
year = {2015},
issue_date = {August 2015},
publisher = {Elsevier Science Inc.},
address = {USA},
volume = {60},
number = {C},
issn = {0167-8655},
url = {https://doi.org/10.1016/j.patrec.2015.04.004},
doi = {10.1016/j.patrec.2015.04.004},
journal = {Pattern Recogn. Lett.},
month = aug,
pages = {24–31},
numpages = {8}
}

@ARTICLE{RKNN-Outlier-detection-TKDE-2014,
  author={Radovanović, Miloš and Nanopoulos, Alexandros and Ivanović, Mirjana},
  journal={IEEE Transactions on Knowledge and Data Engineering}, 
  title={Reverse Nearest Neighbors in Unsupervised Distance-Based Outlier Detection}, 
  year={2015},
  volume={27},
  number={5},
  pages={1369-1382},
  doi={10.1109/TKDE.2014.2365790}}

@article{RNNClustering-2019-InfoSys,
author = {Dai, Qi-Zhu and Xiong, Zhong-Yang and Xie, Jiang and Wang, Xiao-Xia and Zhang, Yu-Fang and Shang, Jia-Xing},
title = {A novel clustering algorithm based on the natural reverse nearest neighbor structure},
year = {2019},
issue_date = {Sep 2019},
publisher = {Elsevier Science Ltd.},
address = {GBR},
volume = {84},
number = {C},
issn = {0306-4379},
url = {https://doi.org/10.1016/j.is.2019.04.001},
doi = {10.1016/j.is.2019.04.001},
journal = {Inf. Syst.},
month = sep,
pages = {1–16},
numpages = {16}
}

@article{KR-DBSCAN-2021-ExpSys,
author = {Hu, Lihua and Liu, Hongkai and Zhang, Jifu and Liu, Aiqin},
title = {KR-DBSCAN: A density-based clustering algorithm based on reverse nearest neighbor and influence space},
year = {2022},
issue_date = {Dec 2021},
publisher = {Pergamon Press, Inc.},
address = {USA},
volume = {186},
number = {C},
issn = {0957-4174},
url = {https://doi.org/10.1016/j.eswa.2021.115763},
doi = {10.1016/j.eswa.2021.115763},
journal = {Expert Syst. Appl.},
month = dec,
numpages = {8}
}

@ARTICLE{RNN-Clustering-TKDE-2017,
  author={Bryant, Avory and Cios, Krzysztof},
  journal={IEEE Transactions on Knowledge and Data Engineering}, 
  title={RNN-DBSCAN: A Density-Based Clustering Algorithm Using Reverse Nearest Neighbor Density Estimates}, 
  year={2018},
  volume={30},
  number={6},
  pages={1109-1121},
  doi={10.1109/TKDE.2017.2787640}}

@inproceedings{Curse-of-dim-1998-ACM,
  title={Approximate nearest neighbors: towards removing the curse of dimensionality},
  author={Indyk, Piotr and Motwani, Rajeev},
  booktitle={Proceedings of the thirtieth annual ACM symposium on Theory of computing},
  pages={604--613},
  year={1998}
}

@article{RDT-2017-VLDB,
author = {Casanova, Guillaume and Englmeier, Elias and Houle, Michael E. and Kr\"{o}ger, Peer and Nett, Michael and Schubert, Erich and Zimek, Arthur},
title = {Dimensional testing for reverse k-nearest neighbor search},
year = {2017},
issue_date = {March 2017},
publisher = {VLDB Endowment},
volume = {10},
number = {7},
issn = {2150-8097},
url = {https://doi.org/10.14778/3067421.3067426},
doi = {10.14778/3067421.3067426},
journal = {Proc. VLDB Endow.},
month = mar,
pages = {769–780},
numpages = {12}
}

@article{RNN-Sampling-2019-PatRecLetter,
author = {Sadhukhan, Payel and Palit, Sarbani},
title = {Reverse-nearest neighborhood based oversampling for imbalanced, multi-label datasets},
year = {2019},
issue_date = {Jul 2019},
publisher = {Elsevier Science Inc.},
address = {USA},
volume = {125},
number = {C},
issn = {0167-8655},
url = {https://doi.org/10.1016/j.patrec.2019.08.009},
doi = {10.1016/j.patrec.2019.08.009},
journal = {Pattern Recogn. Lett.},
month = jul,
pages = {813–820},
numpages = {8}
}

@article{SMOTE-RkNN-2022-Elesevier,
title = {SMOTE-RkNN: A hybrid re-sampling method based on SMOTE and reverse k-nearest neighbors},
journal = {Information Sciences},
volume = {595},
pages = {70-88},
year = {2022},
issn = {0020-0255},
doi = {https://doi.org/10.1016/j.ins.2022.02.038},
url = {https://www.sciencedirect.com/science/article/pii/S0020025522001736},
author = {Aimin Zhang and Hualong Yu and Zhangjun Huan and Xibei Yang and Shang Zheng and Shang Gao}
}

@article{TauMG-2023-SIGMOD,
author = {Peng, Yun and Choi, Byron and Chan, Tsz Nam and Yang, Jianye and Xu, Jianliang},
title = {Efficient Approximate Nearest Neighbor Search in Multi-dimensional Databases},
year = {2023},
issue_date = {May 2023},
publisher = {Association for Computing Machinery},
address = {New York, NY, USA},
volume = {1},
number = {1},
url = {https://doi.org/10.1145/3588908},
doi = {10.1145/3588908},
journal = {Proc. ACM Manag. Data},
month = may,
articleno = {54},
numpages = {27}
}

@ARTICLE{SSSG-2022-TPAMI,
  author={Fu, Cong and Wang, Changxu and Cai, Deng},
  journal={IEEE Transactions on Pattern Analysis and Machine Intelligence}, 
  title={High Dimensional Similarity Search With Satellite System Graph: Efficiency, Scalability, and Unindexed Query Compatibility}, 
  year={2022},
  volume={44},
  number={8},
  pages={4139-4150},
  doi={10.1109/TPAMI.2021.3067706}}

@article{HNSW-2020-TPAMI,
author = {Malkov, Yu A. and Yashunin, D. A.},
title = {Efficient and Robust Approximate Nearest Neighbor Search Using Hierarchical Navigable Small World Graphs},
year = {2020},
issue_date = {April 2020},
publisher = {IEEE Computer Society},
address = {USA},
volume = {42},
number = {4},
issn = {0162-8828},
url = {https://doi.org/10.1109/TPAMI.2018.2889473},
doi = {10.1109/TPAMI.2018.2889473},
journal = {IEEE Trans. Pattern Anal. Mach. Intell.},
month = apr,
pages = {824–836},
numpages = {13}
}

@inproceedings{DiskANN-NIPS-2019,
 author = {Jayaram Subramanya, Suhas and Devvrit, Fnu and Simhadri, Harsha Vardhan and Krishnawamy, Ravishankar and Kadekodi, Rohan},
 booktitle = {Advances in Neural Information Processing Systems},
 editor = {H. Wallach and H. Larochelle and A. Beygelzimer and F. d\textquotesingle Alch\'{e}-Buc and E. Fox and R. Garnett},
 pages = {},
 publisher = {Curran Associates, Inc.},
 title = {DiskANN: Fast Accurate Billion-point Nearest Neighbor Search on a Single Node},
 url = {https://proceedings.neurips.cc/paper_files/paper/2019/file/09853c7fb1d3f8ee67a61b6bf4a7f8e6-Paper.pdf},
 volume = {32},
 year = {2019}
}

@article{Graph-index-survey-VLDB-2021-Wang,
author = {Wang, Mengzhao and Xu, Xiaoliang and Yue, Qiang and Wang, Yuxiang},
title = {A comprehensive survey and experimental comparison of graph-based approximate nearest neighbor search},
year = {2021},
issue_date = {July 2021},
publisher = {VLDB Endowment},
volume = {14},
number = {11},
issn = {2150-8097},
url = {https://doi.org/10.14778/3476249.3476255},
doi = {10.14778/3476249.3476255},
journal = {Proc. VLDB Endow.},
month = jul,
pages = {1964–1978},
numpages = {15}
}

@article{RoarGraph-VLDB-2024-Chen,
author = {Chen, Meng and Zhang, Kai and He, Zhenying and Jing, Yinan and Wang, X. Sean},
title = {RoarGraph: A Projected Bipartite Graph for Efficient Cross-Modal Approximate Nearest Neighbor Search},
year = {2024},
issue_date = {July 2024},
publisher = {VLDB Endowment},
volume = {17},
number = {11},
issn = {2150-8097},
url = {https://doi.org/10.14778/3681954.3681959},
doi = {10.14778/3681954.3681959},
journal = {Proc. VLDB Endow.},
month = jul,
pages = {2735–2749},
numpages = {15}
}

@inproceedings{NN-Descent-WWW-2011,
author = {Dong, Wei and Moses, Charikar and Li, Kai},
title = {Efficient k-nearest neighbor graph construction for generic similarity measures},
year = {2011},
isbn = {9781450306324},
publisher = {Association for Computing Machinery},
address = {New York, NY, USA},
url = {https://doi.org/10.1145/1963405.1963487},
doi = {10.1145/1963405.1963487},
booktitle = {Proceedings of the 20th International Conference on World Wide Web},
pages = {577–586},
numpages = {10},
location = {Hyderabad, India},
series = {WWW '11}
}

@article{Flash-Graph-Index-2025-SIGMOD,
author = {Wang, Mengzhao and Wu, Haotian and Ke, Xiangyu and Gao, Yunjun and Zhu, Yifan and Zhou, Wenchao},
title = {Accelerating Graph Indexing for ANNS on Modern CPUs},
year = {2025},
issue_date = {June 2025},
publisher = {Association for Computing Machinery},
address = {New York, NY, USA},
volume = {3},
number = {3},
url = {https://doi.org/10.1145/3725260},
doi = {10.1145/3725260},
journal = {Proc. ACM Manag. Data},
month = jun,
articleno = {123},
numpages = {29}
}

@article{Graph-tree-index-survey-2023-IEEE-DataEngBull,
  author       = {Zeyu Wang and
                  Peng Wang and
                  Themis Palpanas and
                  Wei Wang},
  title        = {Graph- and Tree-based Indexes for High-dimensional Vector Similarity
                  Search: Analyses, Comparisons, and Future Directions},
  journal      = {{IEEE} Data Eng. Bull.},
  volume       = {47},
  number       = {3},
  pages        = {3--21},
  year         = {2023},
  url          = {http://sites.computer.org/debull/A23sept/p3.pdf},
  bibsource    = {dblp computer science bibliography, https://dblp.org}
}

@article{NSG-VLDB-2019-Fu,
author = {Fu, Cong and Xiang, Chao and Wang, Changxu and Cai, Deng},
title = {Fast approximate nearest neighbor search with the navigating spreading-out graph},
year = {2019},
issue_date = {January 2019},
publisher = {VLDB Endowment},
volume = {12},
number = {5},
issn = {2150-8097},
url = {https://doi.org/10.14778/3303753.3303754},
doi = {10.14778/3303753.3303754},
journal = {Proc. VLDB Endow.},
month = jan,
pages = {461–474},
numpages = {14}
}

@article{AlphaCG-2026-SIGMOD-Lu,
author = {Li, Binhong and Yan, Xiao and Lu, Shangqi},
title = {Fast-Convergent Proximity Graphs for Approximate Nearest Neighbor Search},
year = {2026},
issue_date = {February 2026},
publisher = {Association for Computing Machinery},
address = {New York, NY, USA},
volume = {4},
number = {1},
url = {https://doi.org/10.1145/3786650},
doi = {10.1145/3786650},
journal = {Proc. ACM Manag. Data},
month = apr,
articleno = {36},
numpages = {24}
}

@inproceedings{Graph-Theory-2023-Indyk,
author = {Indyk, Piotr and Xu, Haike},
title = {Worst-case performance of popular approximate nearest neighbor search implementations: guarantees and limitations},
year = {2023},
publisher = {Curran Associates Inc.},
address = {Red Hook, NY, USA},
booktitle = {Proceedings of the 37th International Conference on Neural Information Processing Systems},
articleno = {2891},
numpages = {18},
location = {New Orleans, LA, USA},
series = {NIPS '23}
}

@article{Fast-Graph-2025-SIGMOD-Lu,
author = {Lu, Shangqi and Tao, Yufei},
title = {Proximity Graphs for Similarity Search: Fast Construction, Lower Bounds, and Euclidean Separation},
year = {2025},
issue_date = {November 2025},
publisher = {Association for Computing Machinery},
address = {New York, NY, USA},
volume = {3},
number = {5},
url = {https://doi.org/10.1145/3767716},
doi = {10.1145/3767716},
journal = {Proc. ACM Manag. Data},
month = nov,
articleno = {280},
numpages = {25}
}

@INPROCEEDINGS{LSG-ICDE-2025-Wang,
  author={Wang, Hongya and Wu, Wenlong and Luo, Cong and Bian, Aobei and Meng, Chunguang and Wu, Yishuo and Sun, Ji},
  booktitle={2025 IEEE 41st International Conference on Data Engineering (ICDE)}, 
  title={Boosting Accuracy and Efficiency for Vector Retrieval with Local Scaling Graph}, 
  year={2025},
  volume={},
  number={},
  pages={336-348},
  doi={10.1109/ICDE65448.2025.00032}}

@article{HVS-2021-VLDB,
  title={HVS: hierarchical graph structure based on voronoi diagrams for solving approximate nearest neighbor search},
  author={Lu, Kejing and Kudo, Mineichi and Xiao, Chuan and Ishikawa, Yoshiharu},
  journal={Proceedings of the VLDB Endowment},
  volume={15},
  number={2},
  pages={246--258},
  year={2021},
  publisher={VLDB Endowment}
}

@article{ELPIS-2023-VLDB,
  title={Elpis: Graph-based similarity search for scalable data science},
  author={Azizi, Ilias and Echihabi, Karima and Palpanas, Themis},
  journal={Proceedings of the VLDB Endowment},
  volume={16},
  number={6},
  pages={1548--1559},
  year={2023},
  publisher={VLDB Endowment}
}

@article{LMG-2025-ACMTDS-Xie,
  title={Graph based k-nearest neighbor search revisited},
  author={Xie, Jiadong and Yu, Jeffrey Xu and Liu, Yingfan},
  journal={ACM Transactions on Database Systems},
  volume={50},
  number={4},
  pages={1--30},
  year={2025},
  publisher={ACM New York, NY}
}

@inproceedings{InfluenceZone-2011-ICDE,
author = {Cheema, Muhammad Aamir and Lin, Xuemin and Zhang, Wenjie and Zhang, Ying},
title = {Influence zone: Efficiently processing reverse k nearest neighbors queries},
year = {2011},
isbn = {9781424489596},
publisher = {IEEE Computer Society},
address = {USA},
url = {https://doi.org/10.1109/ICDE.2011.5767904},
doi = {10.1109/ICDE.2011.5767904},
booktitle = {Proceedings of the 2011 IEEE 27th International Conference on Data Engineering},
pages = {577–588},
numpages = {12},
series = {ICDE '11}
}

@article{FINCH-2008-VLDB,
  title={Finch: Evaluating reverse k-nearest-neighbor queries on location data},
  author={Wu, Wei and Yang, Fei and Chan, Chee-Yong and Tan, Kian-Lee},
  journal={Proceedings of the VLDB Endowment},
  volume={1},
  number={1},
  pages={1056--1067},
  year={2008},
  publisher={VLDB Endowment}
}

@inproceedings{TPL-2004-VLDB,
author = {Tao, Yufei and Papadias, Dimitris and Lian, Xiang},
title = {Reverse kNN search in arbitrary dimensionality},
year = {2004},
isbn = {0120884690},
publisher = {VLDB Endowment},
booktitle = {Proceedings of the Thirtieth International Conference on Very Large Data Bases - Volume 30},
pages = {744–755},
numpages = {12},
location = {Toronto, Canada},
series = {VLDB '04}
}

@article{TPL2-2007-Springer,
author = {Tao, Yufei and Papadias, Dimitris and Lian, Xiang and Xiao, Xiaokui},
title = {Multidimensional reverse kNN search},
year = {2007},
issue_date = {Jul 2007},
publisher = {Springer-Verlag},
address = {Berlin, Heidelberg},
volume = {16},
number = {3},
issn = {1066-8888},
url = {https://doi.org/10.1007/s00778-005-0168-2},
doi = {10.1007/s00778-005-0168-2},
journal = {The VLDB Journal},
month = jul,
pages = {293–316},
numpages = {24}
}

@inproceedings{MRKNNCop-2006-SIGMOD,
author = {Achtert, Elke and B\"{o}hm, Christian and Kr\"{o}ger, Peer and Kunath, Peter and Pryakhin, Alexey and Renz, Matthias},
title = {Efficient reverse k-nearest neighbor search in arbitrary metric spaces},
year = {2006},
isbn = {1595934340},
publisher = {Association for Computing Machinery},
address = {New York, NY, USA},
url = {https://doi.org/10.1145/1142473.1142531},
doi = {10.1145/1142473.1142531},
booktitle = {Proceedings of the 2006 ACM SIGMOD International Conference on Management of Data},
pages = {515–526},
numpages = {12},
location = {Chicago, IL, USA},
series = {SIGMOD '06}
}

@inproceedings{SFT-2003,
author = {Singh, Amit and Ferhatosmanoglu, Hakan and Tosun, Ali \c{S}aman},
title = {High dimensional reverse nearest neighbor queries},
year = {2003},
isbn = {1581137230},
publisher = {Association for Computing Machinery},
address = {New York, NY, USA},
url = {https://doi.org/10.1145/956863.956882},
doi = {10.1145/956863.956882},
booktitle = {Proceedings of the Twelfth International Conference on Information and Knowledge Management},
pages = {91–98},
numpages = {8},
location = {New Orleans, LA, USA},
series = {CIKM '03}
}

@inproceedings{GreedySearch-2011-IJCAI,
author = {Hajebi, Kiana and Abbasi-Yadkori, Yasin and Shahbazi, Hossein and Zhang, Hong},
title = {Fast approximate nearest-neighbor search with k-nearest neighbor graph},
year = {2011},
isbn = {9781577355144},
publisher = {AAAI Press},
booktitle = {Proceedings of the Twenty-Second International Joint Conference on Artificial Intelligence - Volume Volume Two},
pages = {1312–1317},
numpages = {6},
location = {Barcelona, Catalonia, Spain},
series = {IJCAI'11}
}

@INPROCEEDINGS{SLICE-2014-ICDE,
  author={Yang, Shiyu and Cheema, Muhammad Aamir and Lin, Xuemin and Zhang, Ying},
  booktitle={2014 IEEE 30th International Conference on Data Engineering}, 
  title={SLICE: Reviving regions-based pruning for reverse k nearest neighbors queries}, 
  year={2014},
  volume={},
  number={},
  pages={760-771},
  doi={10.1109/ICDE.2014.6816698}}

@inproceedings{ERkNN-2005-CIKM,
author = {Xia, Chenyi and Hsu, Wynne and Lee, Mong Li},
title = {ERkNN: efficient reverse k-nearest neighbors retrieval with local kNN-distance estimation},
year = {2005},
isbn = {1595931406},
publisher = {Association for Computing Machinery},
address = {New York, NY, USA},
url = {https://doi.org/10.1145/1099554.1099697},
doi = {10.1145/1099554.1099697},
booktitle = {Proceedings of the 14th ACM International Conference on Information and Knowledge Management},
pages = {533–540},
numpages = {8},
location = {Bremen, Germany},
series = {CIKM '05}
}

@misc{MRQ-2026-VLDB,
      title={Quantization Meets Projection: A Happy Marriage for Approximate k-Nearest Neighbor Search}, 
      author={Mingyu Yang and Liuchang Jing and Wentao Li and Wei Wang},
      year={2026},
      eprint={2411.06158},
      archivePrefix={arXiv},
      primaryClass={cs.DB},
      url={https://arxiv.org/abs/2411.06158}, 
}

@article{ADSampling-2023-SIGMOD,
author = {Gao, Jianyang and Long, Cheng},
title = {High-Dimensional Approximate Nearest Neighbor Search: with Reliable and Efficient Distance Comparison Operations},
year = {2023},
issue_date = {June 2023},
publisher = {Association for Computing Machinery},
address = {New York, NY, USA},
volume = {1},
number = {2},
url = {https://doi.org/10.1145/3589282},
doi = {10.1145/3589282},
journal = {Proc. ACM Manag. Data},
month = jun,
articleno = {137},
numpages = {27}
}

@article{SYMQG-2025-SIGMOD,
author = {Gou, Yutong and Gao, Jianyang and Xu, Yuexuan and Long, Cheng},
title = {SymphonyQG: Towards Symphonious Integration of Quantization and Graph for Approximate Nearest Neighbor Search},
year = {2025},
issue_date = {February 2025},
publisher = {Association for Computing Machinery},
address = {New York, NY, USA},
volume = {3},
number = {1},
url = {https://doi.org/10.1145/3709730},
doi = {10.1145/3709730},
journal = {Proc. ACM Manag. Data},
month = feb,
articleno = {80},
numpages = {26}
}

@INPROCEEDINGS{DDC-2025-ICDE,
  author={Yang, Mingyu and Li, Wentao and Jin, Jiabao and Zhong, Xiaoyao and Wang, Xiangyu and Shen, Zhitao and Jia, Wei and Wang, Wei},
  booktitle={2025 IEEE 41st International Conference on Data Engineering (ICDE)}, 
  title={Effective and General Distance Computation for Approximate Nearest Neighbor Search}, 
  year={2025},
  volume={},
  number={},
  pages={1098-1110},
  doi={10.1109/ICDE65448.2025.00087}}

@article{PipeANN-2026-OSDI,
author = {Guo, Hao and Lu, Youyou},
title = {Achieving Low-Latency Graph-Based Vector Search via Aligning Best-First Search Algorithm with SSD},
year = {2026},
issue_date = {May 2026},
publisher = {Association for Computing Machinery},
address = {New York, NY, USA},
volume = {22},
number = {2},
issn = {1553-3077},
url = {https://doi.org/10.1145/3793926},
doi = {10.1145/3793926},
journal = {ACM Trans. Storage},
month = mar,
articleno = {12},
numpages = {30}
}

@article{DADE-2024-VLDB,
author = {Deng, Liwei and Chen, Penghao and Zeng, Ximu and Wang, Tianfu and Zhao, Yan and Zheng, Kai},
title = {Efficient Data-Aware Distance Comparison Operations for High-Dimensional Approximate Nearest Neighbor Search},
year = {2024},
issue_date = {November 2024},
publisher = {VLDB Endowment},
volume = {18},
number = {3},
issn = {2150-8097},
url = {https://doi.org/10.14778/3712221.3712244},
doi = {10.14778/3712221.3712244},
journal = {Proc. VLDB Endow.},
month = nov,
pages = {812–821},
numpages = {10}
}

@article{Starling-2024-SIGMOD,
author = {Wang, Mengzhao and Xu, Weizhi and Yi, Xiaomeng and Wu, Songlin and Peng, Zhangyang and Ke, Xiangyu and Gao, Yunjun and Xu, Xiaoliang and Guo, Rentong and Xie, Charles},
title = {Starling: An I/O-Efficient Disk Graph Index Framework for High-Dimensional Vector Similarity Search on Data Segment},
year = {2024},
issue_date = {February 2024},
publisher = {Association for Computing Machinery},
address = {New York, NY, USA},
volume = {2},
number = {1},
url = {https://doi.org/10.1145/3639269},
doi = {10.1145/3639269},
journal = {Proc. ACM Manag. Data},
month = mar,
articleno = {14},
numpages = {27}
}

@misc{OptimizationsGraphbasedDiskresident--2026-li,
  title = {I/{{O}} Optimizations for Graph-Based Disk-Resident Approximate Nearest Neighbor Search: A Design Space Exploration},
  shorttitle = {I/{{O}} Optimizations for Graph-Based Disk-Resident Approximate Nearest Neighbor Search},
  author = {Li, Liang and Gong, Shufeng and Yang, Yanan and Wang, Yiduo and Wu, Jie},
  year = 2026,
  month = feb,
  number = {arXiv:2602.21514},
  eprint = {2602.21514},
  primaryclass = {cs},
  publisher = {arXiv},
  doi = {10.48550/arXiv.2602.21514},
  urldate = {2026-02-27},
  archiveprefix = {arXiv},
  langid = {english}
}

@misc{AiSAQ-2025-arxiv,
      title={AiSAQ: All-in-Storage ANNS with Product Quantization for DRAM-free Information Retrieval}, 
      author={Kento Tatsuno and Daisuke Miyashita and Taiga Ikeda and Kiyoshi Ishiyama and Kazunari Sumiyoshi and Jun Deguchi},
      year={2025},
      eprint={2404.06004},
      archivePrefix={arXiv},
      primaryClass={cs.IR},
      url={https://arxiv.org/abs/2404.06004}, 
}

@inproceedings{PEO-2024-ICML,
author = {Lu, Kejing and Xiao, Chuan and Ishikawa, Yoshiharu},
title = {Probabilistic routing for graph-based approximate nearest neighbor search},
year = {2024},
publisher = {JMLR.org},
booktitle = {Proceedings of the 41st International Conference on Machine Learning},
articleno = {1347},
numpages = {19},
location = {Vienna, Austria},
series = {ICML'24}
}

@inproceedings{ADA-NNS-2025-WWW,
author = {Jung, Sungjun and Park, Yongsang and Lee, Haeun and Oh, Young H. and Lee, Jae W.},
title = {Angular Distance-Guided Neighbor Selection for Graph-Based Approximate Nearest Neighbor Search},
year = {2025},
isbn = {9798400712746},
publisher = {Association for Computing Machinery},
address = {New York, NY, USA},
url = {https://doi.org/10.1145/3696410.3714870},
doi = {10.1145/3696410.3714870},
booktitle = {Proceedings of the ACM on Web Conference 2025},
pages = {4014–4023},
numpages = {10},
location = {Sydney NSW, Australia},
series = {WWW '25}
}

@inproceedings{FINGER-2023-WWW,
author = {Chen, Patrick and Chang, Wei-Cheng and Jiang, Jyun-Yu and Yu, Hsiang-Fu and Dhillon, Inderjit and Hsieh, Cho-Jui},
title = {FINGER: Fast Inference for Graph-based Approximate Nearest Neighbor Search},
year = {2023},
isbn = {9781450394161},
publisher = {Association for Computing Machinery},
address = {New York, NY, USA},
url = {https://doi.org/10.1145/3543507.3583318},
doi = {10.1145/3543507.3583318},
booktitle = {Proceedings of the ACM Web Conference 2023},
pages = {3225–3235},
numpages = {11},
location = {Austin, TX, USA},
series = {WWW '23}
}

@inproceedings{
RankRAG-2024-neuips,
title={Rank{RAG}: Unifying Context Ranking with Retrieval-Augmented Generation in {LLM}s},
author={Yue Yu and Wei Ping and Zihan Liu and Boxin Wang and Jiaxuan You and Chao Zhang and Mohammad Shoeybi and Bryan Catanzaro},
booktitle={The Thirty-eighth Annual Conference on Neural Information Processing Systems},
year={2024},
url={https://openreview.net/forum?id=S1fc92uemC}
}

@inproceedings{
RetRobust-2024-ICLR,
title={Making Retrieval-Augmented Language Models Robust to Irrelevant Context},
author={Ori Yoran and Tomer Wolfson and Ori Ram and Jonathan Berant},
booktitle={The Twelfth International Conference on Learning Representations},
year={2024},
url={https://openreview.net/forum?id=ZS4m74kZpH}
}

@inproceedings{
SelfRAG-2024-ICLR,
title={Self-{RAG}: Learning to Retrieve, Generate, and Critique through Self-Reflection},
author={Akari Asai and Zeqiu Wu and Yizhong Wang and Avirup Sil and Hannaneh Hajishirzi},
booktitle={The Twelfth International Conference on Learning Representations},
year={2024},
url={https://openreview.net/forum?id=hSyW5go0v8}
}

@INPROCEEDINGS{SIFT-IVFADC-2011-ICASSP,
  author={Jégou, Hervé and Tavenard, Romain and Douze, Matthijs and Amsaleg, Laurent},
  booktitle={2011 IEEE International Conference on Acoustics, Speech and Signal Processing (ICASSP)}, 
  title={Searching in one billion vectors: Re-rank with source coding}, 
  year={2011},
  volume={},
  number={},
  pages={861-864},
  doi={10.1109/ICASSP.2011.5946540}}

@inproceedings{MSONG-2012-WWW,
author = {McFee, Brian and Bertin-Mahieux, Thierry and Ellis, Daniel P.W. and Lanckriet, Gert R.G.},
title = {The million song dataset challenge},
year = {2012},
isbn = {9781450312301},
publisher = {Association for Computing Machinery},
address = {New York, NY, USA},
url = {https://doi.org/10.1145/2187980.2188222},
doi = {10.1145/2187980.2188222},
booktitle = {Proceedings of the 21st International Conference on World Wide Web},
pages = {909–916},
numpages = {8},
location = {Lyon, France},
series = {WWW '12 Companion}
}

@article{GIST-PQ-2011-TPAMI,
author = {Jegou, Herve and Douze, Matthijs and Schmid, Cordelia},
title = {Product Quantization for Nearest Neighbor Search},
year = {2011},
issue_date = {January 2011},
publisher = {IEEE Computer Society},
address = {USA},
volume = {33},
number = {1},
issn = {0162-8828},
url = {https://doi.org/10.1109/TPAMI.2010.57},
doi = {10.1109/TPAMI.2010.57},
journal = {IEEE Trans. Pattern Anal. Mach. Intell.},
month = jan,
pages = {117–128},
numpages = {12}
}

@inproceedings{MSMARCO-Dataset-2016-CocoNIPS,
  author = {Nguyen, Tri and Rosenberg, Mir and Song, Xia and Gao, Jianfeng and Tiwary, Saurabh and Majumder, Rangan and Deng, Li},
  booktitle = {CoCo@NIPS},
  editor = {Besold, Tarek Richard and Bordes, Antoine and d'Avila Garcez, Artur S. and Wayne, Greg},
  ee = {https://ceur-ws.org/Vol-1773/CoCoNIPS_2016_paper9.pdf},
  publisher = {CEUR-WS.org},
  series = {CEUR Workshop Proceedings},
  title = {MS MARCO: A Human Generated MAchine Reading COmprehension Dataset.},
  url = {http://dblp.uni-trier.de/db/conf/nips/coco2016.html#NguyenRSGTMD16},
  volume = 1773,
  year = 2016
}

@online{technicalreport,
	title = {Technical Report},
	author    = {Wenxuan Xia and Mingyu Yang and Wentao Li and Wei Wang},
	year = {2026},
	url = {https://github.com/gegeji/hrnn/blob/main/technical_report.pdf}
}

@article{VSAG-VLDB-2025-Mingyu,
  title={VSAG: An Optimized Search Framework for Graph-Based Approximate Nearest Neighbor Search},
  author={Zhong, Xiaoyao and Li, Haotian and Jin, Jiabao and Yang, Mingyu and Chu, Deming and Wang, Xiangyu and Shen, Zhitao and Jia, Wei and Gu, George and Xie, Yi and others},
  journal={Proceedings of the VLDB Endowment},
  volume={18},
  number={12},
  pages={5017--5030},
  year={2025},
  publisher={VLDB Endowment}
}

\clearpage

\end{document}